%% file: ms.tex
\documentclass[runningheads]{llncs}

\usepackage{graphicx}
\usepackage{amssymb}
\usepackage{amsmath}
\usepackage{comment}
\usepackage{todonotes}
\usepackage[utf8]{inputenc}
\usepackage{shuffle}
\usepackage{multirow}
\usepackage{listings}
\usepackage{mathabx}
\usepackage{hyperref}

\newcommand{\perm}{\operatorname{perm}}
\newcommand{\stc}{\operatorname{sc}}

\newcommand{\Shuf}{\mathcal{S}huf}
\newcommand{\SE}{\mathcal{S}\mathcal{E}}

\DeclareMathOperator{\lcm}{lcm}

\def\dotminus{\mathbin{\ooalign{\hss\raise1ex\hbox{.}\hss\cr
  \mathsurround=0pt$-$}}} 
 
\spnewtheorem{my_problem}{Problem}{\bfseries}{\itshape}

\usepackage{apxproof}
\theoremstyle{plain}
\newtheoremrep{theorem}{Theorem}%[section]
 \newtheoremrep{proposition}[theorem]{Proposition}
\newtheoremrep{lemma}[theorem]{Lemma}
 \newtheoremrep{claim}[theorem]{Claim}
 \newtheoremrep{conjecture}[theorem]{Conjecture}
 \newtheoremrep{corollary}[theorem]{Corollary}
 \newtheoremrep{definition}[theorem]{Definition}
 
\DeclareFontFamily{U}{bigshuffle}{}
\DeclareFontShape{U}{bigshuffle}{m}{n}{
  <5-8> s*[1.7] shuffle7
  <8->  s*[1.7] shuffle10
}{}
\DeclareSymbolFont{BigShuffle}{U}{bigshuffle}{m}{n}
\DeclareMathSymbol\bigshuffle{\mathop}{BigShuffle}{"001}
\DeclareMathSymbol\bigcshuffle{\mathop}{BigShuffle}{"002}

\begin{document}
\title{The Commutative Closure of Shuffle Languages over Group Languages is Regular}
\titlerunning{Commutative Closure on Shuffle Languages over Group Languages}

\author{Stefan Hoffmann\orcidID{0000-0002-7866-075X}}
\authorrunning{S. Hoffmann}
\institute{Informatikwissenschaften, FB IV, 
  Universit\"at Trier,  Universitätsring 15, 54296~Trier, Germany, 
  \email{hoffmanns@informatik.uni-trier.de}}
\maketitle              % typeset the header of the contribution
\begin{abstract}
 
 We show that the commutative closure combined with the iterated
 shuffle is a regularity-preserving operation on group languages.
 In particular, for commutative group languages,
 the iterated shuffle is a regularity-preserving operation.
 We also give bounds for the size of minimal recognizing automata.
 %of the result of these two operations on a group language.
 Then, we use this result
 %, together with the known result that
 %the commutative closure yields a regular language
 %when applied to group languages, 
 to deduce that the commutative
 closure of any shuffle language over group languages, i.e., a language given by a shuffle expression, i.e., expressions involving 
 shuffle, iterated shuffle, concatenation, Kleene star and union in any order,
 starting with the group languages, always yields a regular language.

%  For the shuffle of two commutative group languages, we give
%  a sharp bound. For applying the shuffle operation to the commutative
%  closure of multiple group languages we give a state bound that is better
%  than applying general bounds on individual operations.
%  To derive our results, we introduce the state label method as a unifying
%  framework, %which is based on decomposing the commutative image
%  %into unary automata.
%  which is based on a generalized commutative image
%  and a decomposition thereof into unary automata.

\end{abstract}

\section{Introduction}

\input{introduction}

\section{Preliminaries and Definitions}

By $\Sigma$ we denote a finite set of symbols, i.e., an \emph{alphabet}. 
By $\Sigma^*$ we denote the set of all words with the concatenation operation. 
The \emph{empty word}, i.e., the word of length zero, is denoted by $\varepsilon$.
If $u \in \Sigma$, by $|u|$ we denote the length of $u$, and if $a \in \Sigma$,
by $|u|_a$ we denote the number of times the letter $a$ appears in $u$.
A language is a subset $L \subseteq \Sigma^*$. For a language $L \subseteq \Sigma^*$,
we set $L^+ = \{ u_1 \cdots u_n \mid \{ u_1, \ldots, u_n \}\subseteq L, n > 0\}$
and $L^* = L^+ \cup \{\varepsilon\}$. By $\mathbb N_0$, we denote the natural numbers with zero.

A finite (complete and deterministic\footnote{Here, only complete and deterministic automata are used, hence just called automata for short.}) \emph{automaton} $\mathcal A = (\Sigma, Q, \delta, q_0, F)$
over $\Sigma$ consists of a finite state set $Q$, a totally defined transition
function $\delta : Q \times \Sigma \to Q$, start state $q_0 \in Q$
and final state set $F \subseteq Q$. The transition function could
be extended to words in the usual way by setting, for $u \in \Sigma^*$, $a \in \Sigma$
and $q \in Q$, $\hat \delta(q, ua) = \delta(\hat \delta(q, u), a)$
and $\hat \delta(q, \varepsilon) = q$. In the following, we will drop the distinction
with $\delta$ and will denote this extension also by $\delta : Q \times \Sigma^* \to Q$.
The language \emph{recognized}, or \emph{accepted},
by $\mathcal A$ is $L(\mathcal A) = \{ u \in \Sigma^* \mid \delta(q_0, u) \in F\}$.

A \emph{permutation automaton} is an automaton such that 
for each letter $a \in \Sigma$, the function $\delta_a : Q \to Q$
given by $\delta_a(q) = \delta(q, a)$ for $q \in Q$
is bijective. We also say that the letter $a$ permutes the state set.
For a given permutation automaton $\mathcal A = (\Sigma, Q, \delta, q_0, F)$
and $a \in \Sigma$, the \emph{order of the letter $a$ in $\mathcal A$}
is the smallest number $n > 0$ such that $\delta(q, a^n) = q$
for all $q \in Q$. This equals
the order of the letter viewed as a permutation on $Q$.
The maximal order of any permutation
is given by Landau's function, which has growth rate $O(\exp(\sqrt{n\log n}))$~\cite{GaoMRY17,Landau1903}.
A language $L \subseteq \Sigma^*$ is a \emph{group language},
if there exists a permutation automaton~$\mathcal A$
such that $L = L(\mathcal A)$.
By $\mathcal G$ we denote the class of group languages.
This class could be also seen as a variety~\cite{Pin86,DBLP:reference/hfl/Pin97}.

We will also use \emph{regular expressions} occasionally, for the definition
of them, and also for a more detailed treatment of the above notions,
we refer to any textbook on formal language theory or theoretical computer science, for example~\cite{HopUll79}.

Let $\Sigma = \{a_1, \ldots, a_k\}$ be the alphabet. The map $\psi : \Sigma^{\ast} \to \mathbb N_0^k$
given by $\psi(w) = (|w|_{a_1}, \ldots, |w|_{a_k})$ is called the \emph{Parikh morphism}~\cite{DBLP:journals/jacm/Parikh66}.
If $L \subseteq \Sigma^*$, we set $\psi(L) = \{ \psi(w) \mid w \in L \}$.
For a given word $w \in \Sigma^{\ast}$, we define $\perm(w) := \{ u \in \Sigma^{\ast} : \psi(u) = \psi(w) \}$.
If $L \subseteq \Sigma^{\ast}$, then the \emph{commutative} (or \emph{permutational}) \emph{closure}
is $\perm(L) := \bigcup_{w\in L} \perm(w)$.
A language is called \emph{commutative},
if $\perm(L) = L$.

\begin{definition} The \emph{shuffle operation}, denoted by $\shuffle$, is defined by
\label{def:shuffle}
%  \begin{align*}
%     u \shuffle v & := \left\{ \begin{array}{ll}
%      \multirow{2}{*}{$x_1 y_1 x_2 y_2 \cdots x_n y_n  \mid$} &  u = x_1 x_2 \cdots x_n, v = y_1 y_2 \cdots y_n, \\ 
%          &   x_i, y_i \in \Sigma^{\ast}, 1 \le i \le n, n \ge 1
%   \end{array} \right\},
%  \end{align*}
 \begin{multline*}
    u \shuffle v  = \{ w \in \Sigma^*  \mid  w = x_1 y_1 x_2 y_2 \cdots x_n y_n 
    \emph{ for some words } \\ x_1, \ldots, x_n, y_1, \ldots, y_n \in \Sigma^*
    \emph{ such that } u = x_1 x_2 \cdots x_n \emph{ and } v = y_1 y_2 \cdots y_n \},
 \end{multline*}
 for $u,v \in \Sigma^{\ast}$ and 
  $L_1 \shuffle L_2  := \bigcup_{x \in L_1, y \in L_2} (x \shuffle y)$ for $L_1, L_2 \subseteq \Sigma^{\ast}$.
\end{definition}

 %The shuffle operation is commutative, associative and distributive with respect
 %to union. We will use these properties without further mention. 
 In writing formulas
 without brackets, we suppose that the shuffle operation binds stronger than the set operations,
 and the concatenation operator has the strongest binding. 
 
  If $L_1, \ldots, L_n \subseteq \Sigma^*$, we set $\bigshuffle_{i=1}^n L_i = L_1 \shuffle \ldots \shuffle L_n$.
 The \emph{iterated shuffle} of $L \subseteq \Sigma^*$ 
 is $L^{\shuffle,*} = \bigcup_{n \ge 0} \bigshuffle_{i=1}^n L$.
 
 \begin{theorem}[Fernau et al. \cite{FerPSV2017}] % todo + andere referenzen aus jantzen und co
 \label{thm:shuffle_properties}
  Let $U,V,W \subseteq \Sigma^*$. Then,
  \begin{enumerate}
  \item $U \shuffle V = V \shuffle U$ (commutative law);
  \item $(U \shuffle V) \shuffle W = U \shuffle (V \shuffle W)$ (associative law);
  \item $U \shuffle (V \cup W)  
   = (U \shuffle V) \cup (U \shuffle W)$ (distributive over union);
% falsch
%   \item $U \shuffle (V \cap W)  
%   = (U \shuffle V) \cap (U \shuffle W)$ (distributive over intersection);
  \item $(U^{\shuffle,*})^{\shuffle,*} = U^{\shuffle,*}$;
  \item $(U\cup V)^{\shuffle,*} = U^{\shuffle,*} \shuffle V^{\shuffle,*}$;
  \item $(U \shuffle V^{\shuffle,*})^{\shuffle,*} = (U \shuffle (U \cup V)^{\shuffle,*}) \cup \{\varepsilon\}$.
  \end{enumerate}
 \end{theorem}

The next
result is taken from~\cite{FerPSV2017} and gives 
equations like $\perm(UV) = \perm(U) \shuffle \perm(V)$
or $\perm(U^*) = \perm(U)^{\shuffle,*}$ for $U,V \subseteq \Sigma^*$.
A \emph{semiring} is an algebraic structure $(S, +, \cdot, 0, 1)$
such that $(S, +, 0)$ forms a commutative monoid, $(S, \cdot, 1)$ is a monoid
and we have $a\cdot (b + c) = a\cdot b + a\cdot c$, $(b+c)\cdot a = b\cdot a + c\cdot a$
and $0 \cdot a = a \cdot 0 = 0$.

\begin{theorem}[Fernau et al.  \cite{FerPSV2017}]
\label{thm:perm_semiring_hom}
 $\perm : \mathcal P(\Sigma^*) \to \mathcal P(\Sigma^*)$
 is a semiring morphism from the semiring
 $(\mathcal P(\Sigma^*), \cup, \cdot, \emptyset, \{\varepsilon\})$, that also respects the iterated catenation resp. iterated shuffle operation,
 to the semiring $(\mathcal P(\Sigma^*), \cup, \shuffle, \emptyset, \{\varepsilon\})$.
\end{theorem}

As $\psi(U \shuffle V) = \psi(UV)$ and $\psi(U^*) = \psi(U^{\shuffle,*})$,
we also find the next result.

\begin{theorem} % dait perm/(SE) = perm(Shuf)
\label{thm:perm_semiring_hom_shuffle_only}
 $\perm : \mathcal P(\Sigma^*) \to \mathcal P(\Sigma^*)$
 is a semiring morphism from the semiring
 $(\mathcal P(\Sigma^*), \cup, \shuffle, \emptyset, \{\varepsilon\})$
 to the semiring $(\mathcal P(\Sigma^*), \cup, \shuffle, \emptyset, \{\varepsilon\})$
 that also respects the iterated shuffle operation.
\end{theorem}

% in jalc für polynomome noch mit einzelwrot perm(u) shuffle perm(gruppe)
% für einzelwort auf APC (submitted) verweisen
% und erwähnen dass für shuffle expressions erweitert

In~\cite{DBLP:journals/iandc/GomezGP13} it was shown that the commutative closure 
is regularity-preserving on $\mathcal G$ using combinatorial arguments.
In~\cite{DBLP:conf/dcfs/Hoffmann20} an automaton was constructed, yielding
explicit bounds for the number of states needed in any recognizing automaton.

\begin{theorem}[\cite{DBLP:conf/dcfs/Hoffmann20}] % gomezpin und meins, mit bound
\label{thm:perm_grp_lang}
 Let $\Sigma = \{a_1, \ldots, a_k\}$
 and $\mathcal A = (\Sigma, Q, \delta, q_0, F)$
 be a permutation automaton.
 Then
 $
  \perm(L(\mathcal A))
 $
 is recognizable by an automaton
 with at most $\left( |Q|^k \prod_{i=1}^k L_i \right)$
 states, where $L_i$ for $i \in \{1,\ldots,k\}$
 denotes the order of~$a_i$. % viewed as a permutation of the state set $Q$.
 Furthermore, the recognizing automaton is computable.
\end{theorem}

\section{Shuffle Languages over Arbitrary Language Classes}

Here, we introduce shuffle languages
%\footnote{Here, as the expressions
%themselves are not our main concern, we do not distinguish between an expression and the language
%it denotes. Actually, we never introduce expressions formally, as it would involve to have a unified naming scheme
%for the starting language class, which are the group languages in our case, but start by defining the closure of a language
%class under certain operations.} 
over arbitrary language classes
and proof a normal form result.

\begin{definition}
 Let $\mathcal L$ be a class of languages.
 \begin{enumerate}
     \item $\SE(\mathcal L)$ is the closure of $\mathcal L$ under shuffle, iterated shuffle, union, concatenation
      and Kleene star.
      
     \item $\Shuf(\mathcal L)$ is the closure of $\mathcal L$ under shuffle, iterated shuffle and union.
 \end{enumerate}
\end{definition}

 For $\mathcal L_{Alp} = \{ \emptyset, \{\varepsilon\} \} \cup \{ \{a\} \mid a \in \Sigma \mbox{ for some alphabet } \Sigma \}$
 and $\mathcal L_{Fin} = \{ L \mid L \subseteq \Sigma^* \mbox{ for some alphabet and $L$ is finite }\}$
 the resulting closures were investigated in \cite{FerPSV2017,DBLP:journals/tcs/Jantzen81,DBLP:journals/tcs/Jantzen85,DBLP:journals/tcs/JedrzejowiczS01}. 
 Note that $\SE(\mathcal L_{Alp}) = \SE(\mathcal L_{Fin})$.
 By Theorem~\ref{thm:perm_semiring_hom}, we can compute a shuffle expression over $\mathcal L_{Alp}$
 for the commutative closure of any regular language by rewriting a regular expression
 and vice versa.
 Hence, the class $\Shuf(\mathcal L_{Alp})$
 equals the commutative closure of all
 regular languages. So, $\Shuf(\mathcal L_{Alp}) \ne \Shuf(\mathcal L_{Fin})$.

\begin{proposition}
\label{prop:shuffle_exp_NF}
 Let $L \in \Shuf(\mathcal L)$. Then, $L$ is a finite union
 of languages of the form
 \[
  L_1 \shuffle \ldots \shuffle L_{k} \shuffle L_{k+1}^{\shuffle,*} \shuffle \ldots \shuffle L_n^{\shuffle,*}
 \]
 with $1 \le k \le n$ and $L_i \in \mathcal L$ for $i \in \{1,\ldots,n\}$
 and this expression is computable.
\end{proposition}
\begin{proof}
% This follows by induction over shuffle expressions representing languages in $\Shuf(\mathcal L)$
% using Theorem~\ref{thm:shuffle_properties}.
 Theorem~\ref{thm:shuffle_properties} provides an inductive proof of Proposition~\ref{prop:shuffle_exp_NF}.
 Note that a similar statement has been shown in~\cite[Theorem 3.1]{DBLP:journals/tcs/Jantzen81}
 for $\Shuf(\mathcal L_{Fin})$.
 However, as we do not assume that $\mathcal L$ is closed under shuffle or 
 union, we only get the form as stated.~\qed
\end{proof}

\begin{remark}
 By Theorem~\ref{thm:shuffle_properties},
 we can write the languages in Proposition~\ref{prop:shuffle_exp_NF} also in the form
 $
 L_1 \shuffle \ldots \shuffle L_{k} \shuffle ( L_{k+1} \cup \ldots \cup L_n )^{\shuffle,*}.
 $
 So, if $\mathcal L$ is closed under union, which is the case for languages from $\mathcal G$
 over a common alphabet,
 we can write the languages in $\Shuf(\mathcal L)$ as a finite union
 of languages of the form $L_1 \shuffle \ldots \shuffle L_{n-1} \shuffle L_n^{\shuffle,*}$
 with $L_1, \ldots, L_n \in \mathcal L$.
\end{remark}

Lastly, with Theorem~\ref{thm:perm_semiring_hom} and Theorem~\ref{thm:perm_semiring_hom_shuffle_only}, we 
show that up to permutational equivalence
$\SE(\mathcal L)$ and $\Shuf(\mathcal L)$
give the same languages.

\begin{proposition}
\label{prop:SE_Shuf_perm}
 Let $\mathcal L$ be any class of languages.
 Suppose $L \in \SE(\mathcal L)$.
 Then, we can compute $L' \in \Shuf(L)$
 such that $\perm(L) = \perm(L')$.
\end{proposition}
\begin{proof} 
By Theorem~\ref{thm:perm_semiring_hom}
 and Theorem~\ref{thm:perm_semiring_hom_shuffle_only},
 we have, for $U, V \subseteq \Sigma^*$,
 $
  \perm(U \shuffle V) = \perm(U) \shuffle \perm(V)
   = \perm(U \cdot V)
 $
 and $\perm(U^{\shuffle,*}) = \perm(U)^{\shuffle, *} = \perm(U^*)$.
 So, inductively, for $L \in \SE(\mathcal G)$,
 by replacing every concatenation
 with the shuffle and every Kleene star
 with the iterated shuffle,
 we find $L' \in \Shuf(\mathcal G)$
 such that $\perm(L) = \perm(L')$.~\qed 
\end{proof} 

\section{The Commutative Closure on $\SE(\mathcal G)$}

By Proposition~\ref{prop:SE_Shuf_perm},
the commutative closure on $\SE(\mathcal L)$
for any language class $\mathcal L$ equals the commutative closure of $\Shuf(\mathcal L)$.
Theorem~\ref{thm:perm_grp_it_shuffle} of this section, stating that the commutative closure
combined with the iterated shuffle is regular,
is the main ingredient in our proof that the commutative closure is regularity-preserving on $\SE(\mathcal G)$ and the most demanding result in this work.

Note that, in general,
   this combined operation does not preserves regularity,
   as shown by $\perm(\{ab\})^{\shuffle, *} = \{ w \in \{a,b\}^* \mid |w|_a = |w|_b \}$.

\begin{toappendix}

First, we collect some results that we will need later.
Thereafter, we introduce the state label method in a more general formulation.
Then, we apply it and give the full proof of Theorem~\ref{thm:perm_grp_it_shuffle}.

   The method is technical and non-trivial, and the reader might notice, what I find quite remarkable,
   that the main difficulty in the results presented in this work
   poses Theorem~\ref{thm:perm_grp_it_shuffle}, for which the state label method was developed. 
   All the other results follow less or more readily.
   
   In the following, we call a \emph{semi-automaton} a tupel $\mathcal A = (\Sigma, Q, \delta)$
   with $\Sigma$ the input alphabet, $Q$ the finite state set
   and transition function $\delta : Q \times \Sigma \to Q$. This is an automaton without a start state and a final state set,
   and all notions that do not explicitly use the start state or the final state set carry over from automata
   to semi-automata.
   
   For a natural number $n \in \mathbb N_0$, we set $[n] = \{ 0,\ldots, n - 1 \}$.
   Also, let $M \subseteq \mathbb N_0$ be some \emph{finite} set. 
By $\max M$ we denote the maximal element in $M$
with respect to the usual order, and we set $\max \emptyset = 0$.
Also for finite $M \subseteq \mathbb N_0 \setminus\{0\}$, i.e., $M$ is finite without zero in it,
by $\lcm M$ we denote the
least common multiple of the numbers in $M$, and set $\lcm \emptyset = 0$.
   
   \subsection{Other Results Needed}

    Here, we state the following result about unary languages, which we will
    need later in this subsection.
    
    \subsubsection{Unary Languages}
    \label{sec:unary}

Let $\Sigma = \{a\}$ be a unary alphabet. In this section we collect some results about unary languages.
Suppose $L \subseteq \Sigma^{\ast}$ is regular
with an accepting complete deterministic automaton $\mathcal A = (\Sigma, S, \delta, q_0, F)$. Then by considering
the sequence of states $\delta(q_0, a^1), \delta(q_0, a^2), \delta(q_0, a^3), \ldots$ we find\footnote{Recall, the
automaton is assumed to be complete.}
numbers $i \ge 0, p > 0$ with $i$ and $p$ minimal such that $\delta(q_0, a^i) = \delta(q_0, a^{i+p})$.
We call these numbers the \emph{index} $i$ and the \emph{period} $p$ of the automaton $\mathcal A$.
Suppose $\mathcal A$ is initially connected, i.e., $\delta(q_0, \Sigma^*) = Q$.
Then $i+p = |S|$ and the states $\{ q_0, \delta(q_0, a), \ldots, \delta(q_0, a^{i-1}) \}$
constitute the \emph{tail} and the states
\[
\{ \delta(q_0, a^i), \delta(q_0, a^{i+1}), \ldots, \delta(q_0, a^{i+p-1} \}
\]
constitute the unique \emph{cycle} of the automaton.
%If $\mathcal A$ is not initially connected, when we speak of the cycle or tail
%of that automaton we nevertheless mean the above sets, despite the automaton graph
%might have more than one cycle or more than one straight path.
When we speak of the cycle, tail, index or period
of an arbitrary unary automaton we nevertheless mean the above sets, even if the automaton is not initially connected and
the automaton graph
might have more than one cycle or more than one straight path.

    \begin{lemma}
    \label{lem:unary_period_divides}
 Let $\mathcal A = (\Sigma, Q, \delta, q_0, F)$ be some unary automaton.
 If $\delta(s, a^k) = s$ for some state $s \in Q$ 
 and number $k > 0$, then $k$ is divided by the period of $\mathcal A$.
\end{lemma}
\begin{proof}
 Let $i$ be the index, and $p$ the period of $\mathcal A$.
 We write $k = np + r$ with $0 \le r < p$.
 First note that $s$ is on the cycle of $\mathcal A$, i.e.,
 $$
  s \in \{ \delta(q_0, a^i), \delta(q_0, a^{i+1}),\ldots, \delta(q_0,a^{i+p-1}) \}
 $$
 as otherwise $i$ would not be minimal. Then if $s = \delta(q_0, a^{i+j})$ for some $0 \le j < p$
 we have $\delta(q_0, a^{i+k}) = 
 \delta(q_0, a^{i+p+k}) = \delta(q_0, a^{i+j+k+(p-j)}) = \delta(q_0,a^{i+j+(p-j)}) = \delta(q_0,a^i)$.
 So $\delta(q_0,a^i) = \delta(q_0,a^{i+k}) = \delta(q_0, a^{i+np + r}) = \delta(q_0, a^{i+r})$
 which gives $r = 0$ by minimality of $p$. $\qed$
\end{proof}

\input{state_label_method_general}

% todo grobes schema am ende beschreiben, besteht aus dem und dem, so und so angewendet
% outline and application of the state label method

\subsection{The State Label Method for Iterated Shuffle}% -- Preliminaries}

So, after we have introduced the state label method, state label maps that are
compatible with an automaton and a general regularity-criterion for permutation automata,
we are ready to give the proof of Theorem~\ref{thm:perm_grp_it_shuffle}.

\input{proof_it_shuffle_prelim}
\end{toappendix}

%
% statt ordnung, dass als remark und wie in it shuffle als kleinste für nerode klassen
\begin{theoremrep}
\label{thm:perm_grp_it_shuffle}
 Let $\Sigma = \{a_1, \ldots, a_k\}$
 and $\mathcal A = (\Sigma, Q, \delta, q_0, F)$
 be a permutation automaton.
 Then
 \[
  \perm(L(\mathcal A)^{\shuffle,*})
 \]
 is recognizable by an automaton
 with at most $\left( |Q|^k \prod_{j=1}^k L_j \right) + 1$
 many states, where $L_j$ for $j \in \{1,\ldots,k\}$
 denotes the order of $a_j$, 
 and this automaton is effectively computable.
\end{theoremrep}
\begin{proofsketch}
 \input{proof_sketch_it_shuffle}
\end{proofsketch}
\begin{proof}

\input{proof_it_shuffle}
\end{proof}

So, with Theorem~\ref{thm:perm_grp_it_shuffle},
we can derive our next result.

\begin{theorem}
 Let $L \in \Shuf(\mathcal G)$.
 Then $\perm(L)$ is effectively regular.
\end{theorem}
\begin{proof}
 By Proposition~\ref{prop:shuffle_exp_NF}, we only need to consider
 languages of the form 
 $
 L_1 \shuffle \ldots \shuffle L_{k} \shuffle L_{k+1}^{\shuffle,*} \shuffle \ldots \shuffle L_n^{\shuffle,*}
 $
 with $L_i \in \mathcal G$.
 By Theorem~\ref{thm:perm_semiring_hom_shuffle_only},
 $\perm( L_1 \shuffle \ldots \shuffle L_{k} \shuffle L_{k+1}^{\shuffle,*} \shuffle \ldots \shuffle L_n^{\shuffle,*} )$ equals
 \begin{align*}
  \perm(L_1) \shuffle \ldots \shuffle \perm(L_{k}) \shuffle \perm(L_{k+1}^{\shuffle,*}) \shuffle \ldots \shuffle \perm(L_n^{\shuffle,*}).
 \end{align*}
 The shuffle is regularity-preserving~\cite{BrzozowskiJLRS16,DBLP:journals/jalc/CampeanuSY02,Ito04},
 where an automaton for it is computable. 
 So, by Theorem~\ref{thm:perm_grp_lang}
 and Theorem~\ref{thm:perm_grp_it_shuffle}
 the above language is effectively regular, where again for the commutative closure
 of a group language an automaton is computable similarly as outlined at the end
 of the proof sketch for Theorem~\ref{thm:perm_grp_lang}.
 Hence, $\perm(L)$ is effectively regular.~\qed
\end{proof}

% remark mit shuffle schranken kann man state complexity bound angeben,
% hinschreiben, aber sehr grob. aber vielleicht 
% für L_1 shuffle L_{n-1} shuffle L_n^* angeben mit state label.

So, with Proposition~\ref{prop:SE_Shuf_perm}
our next result follows.

\begin{theorem}
 Let $L \in \SE(\mathcal G)$.
 Then $\perm(L)$ is effectively regular.
\end{theorem}

\begin{comment} 
\begin{remark}
 As Theorem~\ref{thm:perm_grp_it_shuffle}
 states a bound for the size of a recognizing
 automaton, bounds for the shuffle~\ref{bor,chup},
 and union~\cite{todo} are known, %group languages
 and in the reforming steps for the normal
 form of Proposition~\ref{prop:shuffle_exp_NF}
 the used ``basic'' group languages do not change,
 we can also derive a bound for the size of a recognizing
 automaton of the resulting
 commutative closure of any shuffle expression
 depending on the ``basic'' group languages
 in this expression
 However, as the best bound for the binary shuffle
 is exponential~\cite{}, this bound would
 the quite rough, complicated and depends on the number of unions used in the normal form.
 %I refrain from writing it up.
\end{remark}
\end{comment}

\section{Commutative Group Languages}

%hier bessere shuffle schranke aus CSR einreichung + it shuffle komm sprachen regulär. 
% + projektionen?

   By Theorem~\ref{thm:perm_grp_it_shuffle}, we can deduce
   that for commutative group languages $L \subseteq \Sigma^*$, the iterated
   shuffle is a regularity-preserving operation.
   Also, for a commutative regular language in general, it is easy to see that for
   a minimal automaton $\mathcal A = (\Sigma, Q, \delta, q_0, F)$
   we must have $\delta(q, ab) = \delta(q, ba)$
   for any $q \in Q$ and $a,b \in \Sigma$~\cite{FernauHoffmann2019}.
   Furthermore, if $\mathcal A = (\Sigma, Q, \delta, q_0, F)$ is a minimal permutation automaton
   for a commutative language, then the order
   of each letter $a \in \Sigma$ equals the minimal $n > 0$
   such that $\delta(q_0, a^n) = q_0$.
   For if $q \in Q$, then, by minimality, there exists $u \in \Sigma^*$
   such that $\delta(q_0, u) = q$, which yields
   $\delta(q, a^n) = \delta(\delta(q_0, u), a^n) = \delta(q_0, a^n u) = \delta(\delta(q_0, a^n), u) = 
   \delta(q_0, u) = q$.
   So, combining our observations, we get the next result.
   
   \begin{proposition}
    Let $\Sigma = \{a_1, \ldots, a_k\}$ and $L \subseteq \Sigma^*$
    be a commutative group language with minimal permutation automaton $\mathcal A = (\Sigma, Q, \delta, q_0, F)$
    such that $L = L(\mathcal A)$.
    Then, the iterated shuffle $L^{\shuffle,*}$
    is regular and recognizable
    by an automaton with at most $( |Q|^k \prod_{i=1}^k p_i ) + 1$ many states,
    where $p_i > 0$ is minimal such that
    $\delta(q_0, a_i^{p_i}) = q_0$ for $i \in \{1,\ldots, k\}$.
   \end{proposition}

\section{The $n$-times Shuffle}
  \label{sec:n_times_shuffle}
  
  We just note in passing that the method of proof 
  of Theorem~\ref{thm:perm_grp_it_shuffle}
  could also be adapted to yield
  a bound for the size of a recognizing
  automaton of the $n$-times shuffle combined with
  the commutative closure on group languages
  that is better than applying the bounds from~\cite{BrzozowskiJLRS16,DBLP:journals/jalc/CampeanuSY02,DBLP:conf/dcfs/Hoffmann20}
  individually.
  
  \begin{toappendix}
  \input{proof_n_ary_shuffle_prelim}

  \end{toappendix}
  
    \begin{propositionrep}\label{thm:n_times_shuffle}
     Let $\mathcal A_i = (\Sigma, Q_i, \delta_i, q_i, F_i)$
     for $i \in\{1, \ldots, n\}$ be $n$ permutation automata.
     Then
     $$
      \stc(\perm(L(\mathcal A_1)) \shuffle \ldots \shuffle \perm(L(\mathcal A_n)))
       \le \left( \sum_{i=1}^n Q_i \right)^k \prod_{j=1}^k \lcm(L_j^{(1)}, \ldots, L_j^{(n)})
     $$
     where $L_j^{(i)}$ for $i \in \{1,\ldots, n\}$ and $j \in \{1,\ldots, k\}$
     denotes
     the order of the letter $a_j$ as a permutation on $Q_i$.
    \end{propositionrep}
    \begin{proof}
    \input{proof_n_ary_shuffle}
    \end{proof}

\section{Conclusion}

We have shown that the commutative closure of any shuffle language
over group languages is regular.
However, it is unknown if any shuffle language over the group
languages is a regular languages itself.
As a first step, the question
if the iterated shuffle of a group language is regular
might be investigated.
I conjecture this to be true, but do not know how to proof it for general group languages. % solvalbe groups
% supersolvable, dafür kennt man ja die sprachen?
Observe that merely by noting that the commutative closure is regular, we cannot conclude
that the original language is regular. For example, consider
the non-regular context-free language given by the grammar $G$ over $\{a,b\}$ with rules
\[
 S \to aTaS \mid \varepsilon, \quad
 T \to bSbT \mid \varepsilon.
\]
and start symbol $S$. % non-regular

\begin{proposition}
 The language $L \subseteq \{a,b\}^*$ generated by the above
 grammar $G$ is not regular,
 but its commutative closure
 is regular.
\end{proposition}
\begin{proof}
\begin{enumerate}

\item $L \cap (ab)^*(ba)^* = \{ (ab)^n(ba)^n \mid n \ge 0 \}$.

 \medskip 
 
 It is easy to see
 that $\{ (ab)^n (ba)^n \mid n \ge 0 \} \subseteq L \cap (ab)^*(ba)^*$.
 For the other inclusion, we will first show that if
 %\footnote{By $\to^{\ast}$, we denote
 %the reflexive and transitive closure of the rewriting relation $\to$.}
 \[
  S \to u 
 \] 
 with $u \in (ab)^*(ba)^*$, then $u  = \varepsilon$
 or $S \to abSba \to u$ with $u = abvba$, which
 implies $v \in (ab)^*(ba)^*$.
 So assume $S \to u$ with $u \ne \varepsilon$.
 Then, we must have
 \[
  S \to aTaS \to u,
 \] 
 As, by assumption $u \notin \Sigma^*aa\Sigma^*$,
 we must apply $S \to \varepsilon$ and could not apply $T \to \varepsilon$.
 So, the following steps are necessary
 \begin{equation} \label{eqn:T_expansion}
  S \to aTaS \to aTa \to abSbTa \to u.
 \end{equation}
 Assume we expand $T$ into a non-empty word,
 then 
 \[
  abSbTa \to abSbbSbTa.
 \] 
 As the factor $bb$ occurs at most once in any word from $(ab)^*(ba)^*$,
 the above must expand to $abSbbSba$.
 This, in turn, implies that the first $S$ must expand into a word from $(ab)^*a$.
 However, such a word always contains either an odd number of $a$'s
 or an odd number of $b$'s, and by the production rules, as these letters are always
 introduced in pairs, this is not possible.
 Hence, we cannot expand $T$ in Equation~\eqref{eqn:T_expansion}
 into a non-empty word and we must have $T \to \varepsilon$.
 Then,
 \[
    S \to aTaS \to aTa \to abSba \to u.
 \]
 So, we can write $u = abvba$ with $v \in (ab)^*(ba)^*$.

 Finally, we reason inductively. 
 If $u = \varepsilon$, then $u \in \{ (ab)^n(ba)^n \mid n \ge 0 \}$.
 Otherwise, by the previously shown statement,
 we have $u = abvba$ with $S \to v$ and $v \in (ab)^*(ba)^*$.
 Hence, inductively, we can assume $v = (ab)^n(ba)^n$ for some $n \ge 0$,
 which implies $u = (ab)^{n+1}(ba)^{n+1}$.

\medskip 

\item The generated language is not regular.

\medskip 

 Assume $L$ is regular. Then, with the above result,
 also $\{ (ab)^n(ba)^n \mid n \ge 0 \}$ would be regular.
 However, for the homomorphism $\varphi : \{c,d\}^* \to \{a,b\}^*$
 given by $\varphi(c) = ab$, $\varphi(d) = ba$
 we have $\{ c^n d^n \mid n \ge 0 \} = \varphi^{-1}( \{ (ab)^n(ba)^n \mid n \ge 0 \} )$.
 As the last language is well-known to be not regular,
 and as regular languages are closed under inverse homomorphic mappings,
 the language $\{ (ab)^n(ba)^n \mid n \ge 0 \}$ could not be regular.

\begin{comment}
 We use the pumping lemma for regular languages~\cite{Hopcroft+Ullman/79/Introduction}.
 Suppose it is regular with pumping constant $N$.
 Consider $w = (ab)^N(ba)^N \in L$.
 %Then, by the pumping lemma we
 %have some pumpable factor $v$ in $(ab)^N$.
 By the pumping lemma,
 we can write $w = xyz$
 with $|xy| \le N$, $|y| > 0$
 and $xy^iz \in L$
 for any $i \ge 0$.
 As each production rule adds two symbols or none, every word in $L$ has even length.
 So, $xz \in L$ and $xyz \in L$ have even length,
 which implies that $y$ has even length.
 As $xy$ must be a prefix of $(ab)^N$
 and $y$ has even length, the first symbol
 and the last symbol of $y$ do not equal.
 So, we either have
 \[
  w = (ab)^{r}(ab)^{s}(ab)^{t}(ba)^N
 \]
 with $N = r + s + t$ and $y = (ab)^s$, or
 \[ 
  w = (ab)^{r}a(ab)^{s}b(ab)^{t}(ba)^N
 \]
 with $N - 1 = r + s + t$ and $y = (ba)^y$.
\end{comment}
 
 \medskip
 
 \item The commutative closure of $L$
  is $\{ u \in \{a,b\}^* : |w|_a \equiv 0 \pmod{2}, |w|_b \equiv 0 \pmod{2}, |w|_a \ge \min\{1,|w|_b\} \}$, which is a regular language.
 
 \medskip
    
    We have, for any $n \ge 0$ and $m \ge 0$,
    that $a(bb)^ma(aa)^n \in L$ and $\varepsilon \in L$.
    Also, as each rule introduces the letters
    $a$ or $b$ in pairs, any word in $L$
    has an even number of $a$ and $b$'s
    and as we can only introduce the letter $b$ with
    the non-terminal $T$, which we only can apply after producing at least one $a$, we see that if we have at least
    one $b$, then we need to have at least one $a$.
    Combining these observations yields
    that the commutative closure equals the language
    written above and the defining conditions of this language
    could be realized by automata.
    
\end{enumerate}
 \noindent So, we have shown the claims made in the proposition.~\qed
\end{proof}

\smallskip \noindent \footnotesize
\textbf{Acknowledgement.} I thank the anonymous reviewers who took their time reading through this work.

\bibliographystyle{splncs04}
\bibliography{permbib} 
\end{document}

%% file: introduction.tex
Having applications in regular model checking~\cite{DBLP:journals/iandc/BouajjaniMT07,DBLP:journals/ita/CeceHM08}, or arising naturally in the
theory of traces~\cite{DBLP:books/ws/95/DR1995,DBLP:conf/latin/Sakarovitch92}, one model for parallelism,
the (partial) commutative closure has been extensively studied~\cite{GinsburgSpanier66,GinsburgSpanier66b,Gohon85,DBLP:journals/iandc/GomezGP13,jalc/Hoffmann21,DBLP:conf/dcfs/Hoffmann20,Lvov73,DBLP:journals/ipl/MuschollP96,Redko63}.

% nicht bloß shuffle, auch konkatenation
% und stern und polynomial operations/marked product
% homomorphismen abgeschlossen
% name shuffle nicht bloß
% für shuffle zeigen, mit gleichungen für konkat + stern
%
% überblick wo das gebraucht wird, kurz stete complexity
%
% state label für inverse hom bild [allgemein]
% sigma = \{ delta({q_0}, \varphi(perm(u)) \}
% auch löschende, damit alternativer beweis APC
%
% perm(L_1 shuffle ... L_n^*) direkt, dann kann man obere klasse um inv hom erweitern.
% 

In~\cite{DBLP:journals/iandc/GomezGP13}, the somewhat informal notion of a \emph{robust class}
was introduced, meaning roughly a class\footnote{We relax the condition from~\cite{GomezA08} that it must be a class of regular languages. However, some mechanism to represent the languages from the class
should be available.} closed under some of the usual
operations on languages, such as Boolean
operations, product, star, shuffle, morphism, inverses
of morphisms, residuals, etc.
\begin{comment}
Then, two guiding problems were presented, one of which
the following problem 

was formulated with respect to the (partial) commutation
operation, which were slightly altered by being not formulated for regular languages exclusively,
%\begin{quote}
    % naja, die shuffle expression sind nicht regulär, zumindest weiß ich das nicht..  
    % [...] By a \emph{robust class}, we mean a class of \textbf{regular languages}
    % closed under some of the usual operations on languages, such as Boolean
    % operations, product, star, shuffle, morphism, inverses
    % of morphisms, residuals, etc. [...]
   
   % group langauges ziteiren
   
   \begin{my_problem}
    When is the closure of a language under [partial]
    commutation regular?
   \end{my_problem}
   
%   \begin{my_problem}
%     Are there any robust classes of a languages under [partial] commutation?
%   \end{my_problem}
  
  %https://www.sciencedirect.com/science/article/pii/0304397581900542
  % für die sprachklassen
  %
  % zeige auch closure under hom
%\end{quote}
\end{comment}
 Motivated by two guiding problems formulated in~\cite{DBLP:journals/iandc/GomezGP13},
 we formulate the following slightly altered, but related problems:
 
 \begin{my_problem}
  When is the closure of a language under [partial] commutation regular?
 \end{my_problem} 
 
 \begin{my_problem}
  Are there any robust classes for some common operations such that the commutative closure
  is (effectively) regular?
 \end{my_problem}
 By effectively regular, we mean the stipulation that an automaton
 of the result of the commutation operation
 is computable from a representational scheme for the language class at hand.

 Here, we will investigate the commutation operation on the closure
 of the class (or variety thereof) of group languages
 under union, shuffle, iterated shuffle, concatenation and Kleene star.
 For the class of finite languages, this closure, called the class of \emph{shuffle languages},
 is definable by so called \emph{shuffle expressions}~\cite{FerPSV2017,DBLP:journals/tcs/Jantzen81,DBLP:journals/tcs/Jantzen85,DBLP:journals/tcs/JedrzejowiczS01,DBLP:conf/stoc/Kimura76,Shaw78zbMATH03592960}.
 This is also true in our case, but the atomic expressions are interpreted not as finite languages, but
 as group languages.
 In this sense, we use the term shuffle expressions, or shuffle language, in a wider sense, by allowing
 different atomic languages.
 It will turn out that the commutation operation yields a regular language
 on this class of languages, and it is indeed effectively regular.
 However, I do not know if the languages class itself consists only of regular languages.

 The shuffle and iterated shuffle have been introduced and studied to understand
 the semantics of parallel programs. This was undertaken, as it appears
 to be, independently by Campbell and Habermann~\cite{CamHab74}, by Mazurkiewicz~\cite{DBLP:conf/mfcs/Marzurkiewicz75}
 and by Shaw~\cite{Shaw78zbMATH03592960}. They introduced \emph{flow expressions}, 
 which allow for sequential operators (catenation and iterated catenation) as well
 as for parallel operators (shuffle and iterated shuffle). 
 These operations have been studied extensively, see for example~\cite{FerPSV2017,DBLP:journals/tcs/Jantzen81,DBLP:journals/tcs/Jantzen85,DBLP:journals/tcs/JedrzejowiczS01}.

 The shuffle operation as a binary operation, but not the iterated shuffle,
is regularity-preserving on all regular languages. The size of recognizing automata was investigated in~\cite{DBLP:journals/iandc/BrodaMMR18,BrzozowskiJLRS16,DBLP:journals/jalc/CampeanuSY02,DBLP:journals/jalc/CaronLP20,Hoffmann20,DBLP:conf/cai/Hoffmann19}.

% todo etwas mehr über shuffle operation schreiben?

% auch group hierarchy?
% + beweis ohne state complexity schranke, aber mit umformung

% janzten endliche wörter, jfa für endliche und perm closrue reg

% conclusion, it shuffle von group langauge regulär?

%% file: state_label_method_general.tex
   \subsection{Overview of the State Label Method}

   The state label method was implicitly used in~\cite{DBLP:conf/dcfs/Hoffmann20}
   to give a state complexity bound for the commutative closure
   of a group language, see~\cite{DBLP:conf/dcfs/Hoffmann20} for an intuitive explanation
   and examples in this special case.

   Here, we extract the method of proof from~\cite{DBLP:conf/dcfs/Hoffmann20}
   in a more abstract setting and formulate it independently of any
   automata. Intuitively, we want to describe a commutative language
   by labeling points from $\mathbb N_0^k$ with subsets.
   We call these subsets \emph{state labels}, as in our applications they arise 
   from the states of given automata. Intuitively and very roughly, the method
   could be thought of as both a refined Parikh map for regular languages
   and a power set construction for automata that incorporates the commutativity
   condition. The connection to languages is stated in
   Theorem~\ref{thm:regularity_condition}. In the framework of 
   the state label method we construct unary automata, see Definition~\ref{def:sequ_grid_decomp_aut},
   that are used to decompose the state label map, see Proposition~\ref{prop:grid_aut_decomp}.
   
   Please also see the proof sketch of Theorem~\ref{thm:perm_grp_it_shuffle}
   supplied in the main text to get a bird's-eye view of the method applied to our situation.

   \subsubsection{Outline of the Method and How to Apply It}
   
   Before giving the formal definitions, let us give a rough outline of the method and how to apply it.
   First, the actual labeling is computed by using another function $f$ that operates on the subsets
   of a set $Q$, which is, in our applications,
   related to the states of one or more automata.
   This other function gives us flexibility in the way the other
   state labels are formed. Usually, the labels are formed by the transitions function(s)
   of one or more automata and additional operations, like adding a start state when a condition is meet.
   The state label function itself, which will be called $\sigma$, uses $f$,
   and computes the state labels out of the predecessor state labels,
   where a predecessor of a state label at a point is a state label that corresponds
   to a point strictly smaller, in the componentwise order, than the point in question.
   For a given automaton $\mathcal A$, we will introduce
   functions $f$ that are \emph{compatible with $\mathcal A$}, a term made precise later,
   and we will show that for permutation automata
   and arbitrary state labels $\sigma$ induced by function compatible with the automaton,
   the commutative closure is regular.
   In summary, the application of the state label method, to show regularity
   of the commutative closure for operations on permutation automata, consists in the following
   steps:
   \begin{enumerate}
   \item \emph{Define a state labeling $\sigma$} with the help of a function $f$ 
    that arises out of the operation and automata in question.
    
   \item \emph{Link the state labeling to the Parik image of the operation}.
    More precisely, determine the resulting state labeling more precisely, and 
    show that the inverse image of a suitable chosen subset of state labels
      equals the Parikh image of the operation in question applied to the language
      of the permutation automaton. This yields that the state label map
      could be used to describe the commutative closure.
      
   \item \emph{Apply regularity conditions}, i.e., apply Proposition~\ref{prop:state_label_for_perm_aut}
    and Theorem~\ref{thm:regularity_condition} to deduce regularity
    and a bound for the size of an automaton.
   \end{enumerate}

   %\subsubsection{Formal Introduction to the State Label Method} 

   \subsection{The State Label Method}
   \label{sec:state_label_method}
   
    Our first definition in this section will be the 
    notion of a state label map.

   \begin{definition}\label{def:state_label_fn}
      Let $\Sigma = \{a_1, \ldots, a_k\}$ and $Q$ be a finite set.
      A \emph{state label function}
      is a function $\sigma : \mathbb N_0^k \to \mathcal P(Q)$
      %such that $\sigma_{\mathcal A}(0, \ldots, 0) = \{ s_0 \}$
      %which could be expressed recursively in terms of 
      %its immediate predecessor points.
      %More formally, if we have a 
      given by another function 
      $f : \mathcal P(Q) \times \Sigma \to \mathcal P(Q)$
      so that
      \begin{equation}\label{eqn:state_label_fn}
       \sigma(p) = \bigcup_{\substack{(q,b) \\ p = q + \psi(b)}} f(\sigma(q), b)
      \end{equation}
      for $p \ne (0,\ldots, 0)$ and
      $\sigma(0,\ldots,0) \in \mathcal P(Q)$ is arbitrary.
      
      % oder f : [n] x sigma -> [n] und dann erweitern? dann im wesentlichen
      % aber automaton, kann man dann diese vereinigung mit reinkodieren?
      % f(s, a) = delta(s,a) \cup delta_B(s_0,a) falls s \in F_A geht dann nicht
      % vielleicht alles als bemerkungen
   \end{definition} % damit schließlich periodic in richtungen, da im wesentlichen 
   % anwendung von b in diese richtungen.
   
   In this context, we call the elements from $Q$ states, even if 
   they do not correspond to an automaton.
   The function $f : \mathcal P(Q) \times \Sigma \to \mathcal P(Q)$
   could be extended to words by setting $f(S, \varepsilon) = S$
   and $f(S, ux) = f(f(S,u),x)$. With this extension
   the next equation could be derived.
   
   \begin{lemma}
   \label{lem:state_label_ind_form}
    Let $\sigma : \mathbb N_0^k \to \mathcal P(Q)$
    be a state label function given by $f : \mathcal P(Q) \times \Sigma \to \mathcal P(Q)$
    and $p = (p_1, \ldots, p_k) \in \mathbb N_0^k$.
    If $1 \le n \le p_1 + \ldots + p_k$, then
    $$
     \sigma(p) = \bigcup_{\substack{ (q, w) \in \mathbb N_0^k\times \Sigma^n \\ p = q + \psi(w)}} f(\sigma(q), w).
    $$
   \end{lemma}
      \begin{proof}
    For $n = 1$ this is simply Definition~\ref{def:state_label_fn}, where $p \ne (0,\ldots,0)$
    by the assumptions. For $n > 1$, by Definition~\ref{def:state_label_fn},
    \begin{align*}
        \sigma(p) & = \bigcup_{\substack{ (q, b) \in \mathbb N_0^k \times \Sigma \\ p = q + \psi(b)}}
                       f(\sigma(q), b).
    \end{align*}                       
    If $b \in \Sigma$, then $p = q + \psi(b)$ implies
    $q_1 + \ldots + q_k = p_1 + \ldots + p_k - 1$.
    Hence $1 \le n - 1 \le q_1 + \ldots + q_k$ and, as inductively 
    $$
    \sigma(q) = \bigcup_{\substack{ (q', u) \in \mathbb N_0^k\times \Sigma^{n-1} \\ q = q' + \psi(u)}}  f(\sigma(q'), u),
    $$
    we get
    \begin{align*} 
     \sigma(p) & = \bigcup_{\substack{ (q, b) \in \mathbb N_0^k \times \Sigma \\ p = q + \psi(b)}}
                       f\left(  \bigcup_{\substack{ (q', u) \in \mathbb N_0^k\times \Sigma^{n-1} \\ q = q' + \psi(u)}}  f(\sigma(q'), u), b \right) \\ 
               & = \bigcup_{\substack{ (q, b) \in \mathbb N_0^k \times \Sigma \\ p = q + \psi(b)}}
                   \bigcup_{\substack{ (q', u) \in \mathbb N_0^k\times \Sigma^{n-1} \\ q = q' + \psi(u)}}                   f(f(\sigma(q'), u), b) \\
               & = \bigcup_{\substack{ (q, u) \in \mathbb N_0^k\times \Sigma^{n-1}, b \in \Sigma \\ p = q + \psi(u) + \psi(b)}} f(f(\sigma(q), u), b) \\
               & = \bigcup_{\substack{ (q, u) \in \mathbb N_0^k\times \Sigma^{n-1}, b \in \Sigma \\ p = q + \psi(u) + \psi(b)}} f(\sigma(q), ub) \\ 
               & = \bigcup_{\substack{ (q, w) \in \mathbb N_0^k\times \Sigma^n \\ p = q + \psi(w)}} f(\sigma(q), w).
    \end{align*}
    So, the formula holds true. $\qed$
   \end{proof}
 
   %Now we will connect the state labelling to unary automata.
   %Our method consists of decomposing the labelling into unary automata
   %along rays in the direction of a given letter.
   
   Next, we introduce the hyperplanes that will be used in Definition~\ref{def:sequ_grid_decomp_aut}.

    \begin{definition}{(hyperplane aligned with letter)}\label{def:hyperplance}
     Let $\Sigma = \{a_1, \ldots, a_k\}$ and $j \in \{1,\ldots, k\}$.
     We set
    % \begin{equation}\label{eqn:hyperplane}
     $H_j = \{ (p_1, \ldots, p_k) \in \mathbb N_0^k \mid p_j = 0\}.$
    % \end{equation}
    \end{definition}
    
    Suppose $\Sigma = \{a_1, \ldots, a_k\}$ and $j \in \{1,\ldots, k\}$.
    We will decompose the state label map into unary automata.
    For each letter $a_j$ and point $p \in H_j$, we construct unary automata $\mathcal A_p^{(j)}$.
    They are meant to read inputs in the direction $\psi(a_j)$, which is orthogonal
    to $H_j$. This will be stated more precisely in Proposition \ref{prop:grid_aut_decomp}.

   \begin{definition}{(unary automata along letter $a_j \in \Sigma$)}
    \label{def:sequ_grid_decomp_aut}
     Let $\Sigma = \{a_1, \ldots, a_k\}$ and $\sigma : \mathbb N_0^k \to \mathcal P(Q)$
     be a state label function, with defining
     function $f : \mathcal P(Q) \times \Sigma \to \mathcal P(Q)$
     and finite set $Q$.
     Fix  $j \in \{1,\ldots, k\}$ and $p \in H_j$.
     We define a unary automaton $\mathcal A_p^{(j)} = (\{a_j\}, Q_p^{(j)}, \delta_p^{(j)}, s_p^{(0,j)}, F_p^{(j)})$. But suppose for points $q \in \mathbb N_0^k$
     with $p = q + \psi(b)$ for some $b \in \Sigma$ 
     the unary automata $\mathcal A_q^{(j)} = (\{a_j\}, Q_q^{(j)}, \delta_q^{(j)}, s_q^{(0,j)}, F_q^{(j)})$ are already defined.
     Set\footnote{Note that in the definition of $\mathcal P$,
     as $p \in H_j$, we have $b \ne a_j$ and $q \in H_j$. In general, points $q \in \mathbb N_0^k$
     with $p = q + \psi(b)$ for some $b \in \Sigma$ are  predecessor points in the grid~$\mathbb N_0^k$.}
     $$
      \mathcal P =\{ \mathcal A_q^{(j)} \mid p = q + \psi(b) \mbox{ for some } b \in \Sigma \}.
     $$ %final definieren, aber remark dass nicht so wichtig, oder sprachen angucken? (quotient 
     % nach einem wort und geschnitten den buchstaben?
     Let $I$ be the maximal index 
     and $P$ the least common multiple\footnote{Note $\max\emptyset = 0$
     and $\lcm\emptyset = 1$.} of the periods of the unary automata in $\mathcal P$.
     %Define $\mathcal A_p^{(j)} = (\{a_j\}, Q_p^{(j)}, \delta_p^{(j)}, s_p^{(0,j)}, F_p^{(j)})$,
    % with
     Then set
     \begin{align} 
        Q_p^{(j)} & = \mathcal P(Q) \times [I+P], \nonumber  \\
        s_p^{(0,j)} & = (\sigma(p),0), \label{eqn:def_unary_aut_start_state} \\ 
        \label{eqn:def_unary_aut_transition}
        \delta_p^{(j)}( (S, i), a_j ) & =
       \left\{ \begin{array}{ll}
         (T, i + 1) & \mbox{ if } i+1 < I+P; \\ 
         (T, I)     & \mbox{ if } i+1 = I+P;
       \end{array}\right.
    \end{align}
     where $S \subseteq Q$, $i \in [I + P]$, $j \in \{1,\ldots,k\}$,
    \begin{equation}\label{eqn:def_seq_grid_aut_transition}
      T = f(S, a_j)
           \cup \bigcup_{\substack{(q, b) \in \mathbb N_0^k \times\Sigma\\ p = q + \psi(b)}} f(\pi_1(\delta_q^{(j)}(s_q^{(0,j)}, a_j^{i+1})), b)
    \end{equation}
    and $F_p^{(j)} = \{ (S, i) \mid S \cap F \ne \emptyset \}$.
    For a state $(S, i) \in Q_p^{(j)}$ the set $S \subseteq Q$ will be called
    the \emph{state (set) label}, or the \emph{state set associated with it}.
    \end{definition}
   
    The reader might consult~\cite{DBLP:conf/dcfs/Hoffmann20} for examples.
    The next statement makes
    precise what we mean by decomposing the state label map along the 
    hyperplanes into the automata $\mathcal A_p^{(j)} = (\{a_j\}, Q_p^{(j)}, \delta_p^{(j)}, s_p^{(0,j)}, F_p^{(j)})$.
    %for $p \in H_j$, $j \in \{1,\ldots, k\}$.
    Moreover, it justifies calling the first component of any state $(S, i) \in Q_p^{(j)}$
    also the state set label.
    
    \begin{proposition}{(state label map decomposition)} \label{prop:grid_aut_decomp}
     Suppose $\Sigma = \{a_1, \ldots, a_k\}$ and $Q$ is a finite set.
     Let $\sigma : \mathbb N_0^k \to \mathcal P(Q)$ be a state label map, 
     $1 \le j \le k$ and $p = (p_1, \ldots, p_k) \in \mathbb N_0^k$.
     Assume $\overline p \in H_j$ is the projection of $p$ onto $H_j$, i.e., $\overline p = (p_1, \ldots, p_{j-1}, 0, p_{j+1}, \ldots, p_k)$.
     Then
     $$
      \sigma(p) = \pi_1(\delta_{\overline p}^{(j)}( s_{\overline p}^{(0,j)}, a_j^{p_j} ))
     $$
     for the automata $\mathcal A_{\overline p}^{(j)} = (\{a_j\}, Q_{\overline p}^{(j)}, \delta_{\overline p}^{(j)}, s_{\overline p}^{(0,j)}, F_{\overline p}^{(j)})$ from Definition \ref{def:sequ_grid_decomp_aut}.
    \end{proposition}
 \begin{proof} % induktiven argument noch hervorheben todo
      Notation as in the statement. Also, let $f : \mathcal P(Q) \times \Sigma \to \mathcal P(Q)$
      be the defining function for the state label map.
      For $p = (0, \ldots, 0)$ this is clear.
      If $p_j = 0$, then $p = \overline p$, and, by Equation~\eqref{eqn:def_unary_aut_start_state},
      $$
       \pi_1(\delta_{\overline p}^{(j)}(s_{\overline p}^{(0,j)}, \varepsilon))
        = \pi_1( s_{\overline p}^{(0,j)} ) = \sigma(\overline p).
      $$
      Suppose $p_j > 0$ from now on.
      %By Equation \eqref{eqn:state_label_fn_ind_form}
      %and inductively,
      %we have
      Then, the set $\{ (q,b) \in \mathbb N_0^k \times \Sigma \mid p = q + \psi(b) \}$
      is non-empty and we can use Equation \eqref{eqn:state_label_fn} and,
      inductively, that
      $
       \sigma(q) = \pi_1(\delta_{\overline q}^{(j)}( s_{\overline q}^{(0,j)}, a_j^{q_j} )),
      $
      which gives
      \begin{align}
         \sigma(p) & =  \bigcup_{ \substack{ (q,b) \\ p = q + \psi(b) }} f( \sigma(q), b ) \nonumber \\ 
                              & =  \bigcup_{ \substack{ (q,b) \\ p = q + \psi(b) }} f( \pi_1(\delta_{\overline q}^{(j)}( s_{\overline q}^{(0,j)}, a_j^{q_j} )), b ) \label{eqn:ind_hyp_decomp_proof}
      \end{align}
      where $q = (q_1, \ldots, q_k)$ and $\overline q = (q_1, \ldots, q_{j-1},0,q_{j+1}, \ldots, q_k) \in H_j$.
      As $p_j > 0$ we have $p = q + \psi(a_j)$ for some unique point $q = (p_1, \ldots, p_{j-1}, p_j - 1, p_{j+1} \ldots, p_k)$. 
      For all other points $r = (r_1, \ldots, r_k)$ with $p = r + \psi(b)$ for some $b\in \Sigma$,
      the condition $r \ne q$ implies $b \ne a_j$ and $r_j = p_j$
      for $r = (r_1, \ldots, r_k)$. Also, if $\overline q \in H_j$ denotes the projection to $H_j$, we have $\overline q  = \overline p$
      for our chosen $q$ with $p = q + \psi(a_j)$.
      Hence, taken all this together, we can write Equation \eqref{eqn:ind_hyp_decomp_proof}
      in the form %ncoh hinschreiben (q,b) type bei allen anderen steht nicht dabei weil klar
      \begin{equation*}
       \sigma_{\mathcal A}(p) = \left( %\bigcup_{ \substack{ (r,b) \in \mathbb N_0^k \times \Sigma\setminus\{a_j\} \\ p = r + \psi(b) }}
          \bigcup_{ \substack{ (r,b), b \ne a_j \\ p = r + \psi(b) }}
          f( \pi_1(\delta_{\overline r}^{(j)}( s_{\overline r}^{(0,j)}, a_j^{p_j} )), b ) \right)
          \cup 
          f( \pi_1(\delta_{\overline p}^{(j)}( s_{\overline p}^{(0,j)}, a_j^{p_j-1})), a_j).
      \end{equation*}
      Let $b \in\Sigma$. As for $a_j \ne b$, we have that $p = r + \psi(b)$ if and only if $\overline p = \overline r + \psi(b)$,
      with the notation as above for $p, r, \overline p$ and $\overline r = (r_1, \ldots, r_{j-1},0,r_{j+1}, \ldots, r_k)$, we can simplify further and write
      \begin{equation}\label{eqn:decomp_proof_psi_of_p}
       \sigma_{\mathcal A}(p) = \left( %\bigcup_{ \substack{ (r,b) \in \mathbb N_0^k \times \Sigma\setminus\{a_j\} \\ p = r + \psi(b) }}
          \bigcup_{ \substack{ (\overline r,b), \overline r \in H_j \\ \overline p = \overline r + \psi(b) }}
          f( \pi_1(\delta_{\overline r}^{(j)}( s_{\overline r}^{(0,j)}, a_j^{p_j} )), b ) \right)
          \cup 
          f( \pi_1(\delta_{\overline p}^{(j)}( s_{\overline p}^{(0,j)}, a_j^{p_j-1})), a_j).
      \end{equation}

      \noindent Set $S = \pi_1(\delta_{\overline p}^{(j)}( s_{\overline p}^{(0,j)}, a_j^{p_j-1}))$,
      $T =  \sigma(p)$ and\footnote{Note that for $\overline p \in H_j$, the condition $\overline p = q + \psi(b)$, for some $b \in \Sigma$,
      implies $q \in H_j$ and $b \ne a_j$.}
      $$
       \mathcal P = \{ \mathcal A_r^{(j)} \mid \overline p = r + \psi(b) \mbox{ for some } b \in \Sigma \}. 
      $$
      Let $I$ be the maximal index, and $P$ the least common multiple of all the periods, of
      the unary automata in $\mathcal P$. 
      We distinguish two cases for the value of $p_j > 0$.
      
      \begin{enumerate} 
      \item[(i)] $0 < p_j \le I$.
        
        By Equation \eqref{eqn:def_unary_aut_transition}, $\delta_{\overline p}( s_{\overline p}^{(0,j)}, a_j^{p_j-1} )= (S, p_j-1)$.
        In this case Equation \eqref{eqn:decomp_proof_psi_of_p}
        equals Equation \eqref{eqn:def_seq_grid_aut_transition}, if the state $(S, p_j-1)$ is used
        in Equation \eqref{eqn:def_unary_aut_transition}, i.e.,
        $$
         T =
          f( S, a_j)  \cup  \left( 
          \bigcup_{ \substack{ (\overline r,b), \overline r \in H_j \\ \overline p = \overline r + \psi(b) }}
          f( \pi_1(\delta_{\overline r}^{(j)}( s_{\overline r}^{(0,j)}, a_j^{p_j} )), b ) \right).
        $$
        This
        gives 
        $$
         \delta_{\overline p}^{(j)}( (S, p_j-1), a_j ) = (T, p_j).
        $$
        Hence $\pi_1(\delta_{\overline p}^{(j)}( (S, p_j-1), a_j )) = T = \sigma(p)$.
      
      \item[(ii)] $I < p_j$.
      
        Set $y = I + ((p_j - 1 - I) \bmod P)$.
        Then $I \le y < I + P$.
        By Equation~\eqref{eqn:def_unary_aut_transition},
        $
           \delta_{\overline p}^{(j)}( s_{\overline p}^{(0,j)}, a_j^{p_j-1} ) = ( S, y ).
        $
        So, also by Equation~\eqref{eqn:def_unary_aut_transition},
        $$ 
          \delta_{\overline p}^{(j)}( s_{\overline p}^{(0,j)}, a_j^{p_j} ) = 
          \delta_{\overline p}^{(j)}( (S,y), a_j ) = 
          \left\{ 
           \begin{array}{ll}
            (R, y + 1) & \mbox{ if } I \le y < I + P - 1 \\ 
            (R, I)     & \mbox{ if } y = I + P - 1,
           \end{array}
          \right.
        $$
        where, by Equation~\eqref{eqn:def_seq_grid_aut_transition},
        \begin{equation}\label{eqn:R}
         R = f(S, a_j) \cup \bigcup_{\substack{(\overline r, b) \\ \overline p = \overline r + \psi(b)}} f(\pi_1(\delta_{\overline r}^{(j)}(s_{\overline r}^{(0,j)}, a_j^{y+1})), b).
        \end{equation}
        Let $\overline r \in H_j$ with $\overline p = \overline r + \psi(b)$
        for some $b \in \Sigma$, and $\overline p \in H_j$ the point from the statement of this Proposition.
        Then, as the period of $\mathcal A_{\overline r}^{(j)}$ divides $P$, and $y$ is greater than or equal to
        the index of $\mathcal A_{\overline r}^{(j)}$, we have
        $$
          \delta_{\overline r}^{(j)}(s_{\overline r}^{(0,j)}, a_j^{p_j-1}) = \delta_{\overline r}^{(j)}(s_{\overline r}^{(0,j)}, a_j^{y}).
        $$
        So $\delta_{\overline r}^{(j)}(s_{\overline r}^{(0,j)}, a_j^{p_j}) = \delta_{\overline r}^{(j)}(s_{\overline r}^{(0,j)}, a_j^{y+1})$.
        Hence, comparing Equation \eqref{eqn:R} with Equation \eqref{eqn:decomp_proof_psi_of_p},
        we find that they are equal, and so $R = T$. $\qed$ % todo nebensatzstruktur kommata
     \end{enumerate}
      \end{proof}

   By Propositon~\ref{prop:grid_aut_decomp},
   the state label sets of the axis-parallel rays in $\mathbb N_0^k$
   correspond to the state set labels of unary automata. Hence, the next
   is implied.

   \begin{corollary}
    A state label map is ultimately periodic along each ray. More formally,
    if $\sigma : \mathbb N_0^k \to \mathcal P(Q)$ 
    is a state label function, $p \in \mathbb N_0^k$ and $j \in \{1,\ldots, k\}$, then
    the sequence of state sets $\sigma(p + i \cdot \psi(a_j))$
    for $i = 0, 1, 2, \ldots$ is ultimately periodic.
   \end{corollary}
   
   %\begin{theorem}
    %wenn beschränkt, dann für jede auswähl von teilmengen, die wörter
    %welche diese als bilder haben ist reguläre menge.
   %\end{theorem}
   
    Our final result in this section is the mentioned regularity
    condition. It says that if
    the automata from Definition~\ref{def:sequ_grid_decomp_aut} underlying the state set labels, as stated in Proposition~\ref{prop:grid_aut_decomp},
    do not grow, i.e., have a bounded number of states,
    then we can deduce that the languages we get if we look at the inverse
    images of the state label map and the Parikh map are regular.
    This is equivalent with the condition that the state set labels all get periodic
    behind specific points, i.e., outside of some bounded rectangle
    in $\mathbb N_0^k$.

    \begin{theorem}\label{thm:regularity_condition}
     %Let $\mathcal A = (\Sigma, Q, \delta, s_0, F)$
     %be some finite automaton. 
     
     Let $\sigma : \mathbb N_0 \to \mathcal P(Q)$ be a state label map and $\psi : \Sigma^* \to \mathbb N_0^k$
     be the Parikh map.
     Suppose for every $j \in \{1,\ldots, k\}$ and $p \in H_j$ the automata $\mathcal A_p^{(j)} = (\{a_j\}, Q_p^{(j)}, \delta_p^{(j)}, s_p^{(0,j)}, F_p^{(j)})$
     from Definition \ref{def:sequ_grid_decomp_aut}
     have a bounded number of states\footnote{Equivalently, the index and period is bounded, 
     which is equivalent %todo isomorphism defniieren?
     with just a finite number of distinct automata, up to semi-automaton isomorphism.
     We call two semi-automata isomorphic if one semi-automaton can be obtained from the other one by renaming states and alphabet symbols.}, i.e., $|Q_p^{(j)}| \le N$ for some $N \ge 0$ independent of $p$ and $j$.
     %, even with respect to the state labels 
     %by the previous corollary), naja, etwas schwammig, gibt ja auch nur endliche teilmengen, also wiederholt sich das auch
     %Then the commutative closure $\perm(L(\mathcal A))$     is regular
     %and could be accepted by an automaton of size
     %Then, for %each family of subsets of states 
     %$\mathcal F \subseteq \mathcal P(Q)$,
     Then for $\mathcal F \subseteq \mathcal P(Q)$
     the commutative language
     $$
      \psi^{-1}(\sigma^{-1}( \mathcal F ))
     $$
     is regular and could be accepted by an automaton of size
     $
      \prod_{j=1}^k (I_j + P_j),
     $
     where $I_j$ denotes the largest index among
     the unary automata
     $
      \{ \mathcal A_p^{(j)} \mid p \in H_j \}
     $
     and $P_j$ the least common multiple
     of all the periods of these automata.
     In particular, by the relations of the index and period to the states from Section \ref{sec:unary},
     the state complexity of $\psi^{-1}(\sigma^{-1}( \mathcal F ))$ is bounded by $N^k$.
    \end{theorem} 
    \begin{proof}
    %auf mein paper veweisen für diese konstruktion
     %For each $1 \le i \le k$ let $I_i$ be the largest index among
     %the automata $\mathcal A_p^{(i)}$, and $P_i$ the least common multiple
     %of all the periods. 
     We use the same notation as introduced in the statement of the theorem.
     Let $p = (p_1, \ldots, p_k) \in \mathbb N_0^k$ and $j \in \{1,\ldots k\}$.
     Denote by $\sigma : \mathbb N_0^k \to \mathcal P(Q)$ the state label function
     from Definition~\ref{def:state_label_fn}.
     By Proposition \ref{prop:grid_aut_decomp},
     if $p_j \ge I_j$, we have 
     \begin{equation}\label{eqn:state_label_fn_ult_periodic}
     \sigma(p_1, \ldots, p_{j-1}, p_j + P_j, p_{j+1}, \ldots, p_k)
      = \sigma(p_1, \ldots, p_k).
     \end{equation}
     Construct the unary semi-automaton\footnote{The term semi-automaton is used
     for automata without a designated initial state, nor a set of final states.}
     $\mathcal A_j = (\{a_j\}, Q_j, \delta_j)$
     with 
     \begin{align*}
        Q_j & = \{s^{(j)}_0, s_1^{(j)}, \ldots, s_{I_j + P_j - 1}^{(j)} \}, \\
        \delta_j( s_i^{(j)} , a_j ) & = \left\{ \begin{array}{ll}
       s_{i+1}^{(j)} & \mbox{ if } i < I_j \\ 
       s_{I_j + (i-I_j+1) \bmod P_j}^{(j)} & \mbox{ if } i \ge I_j. \end{array}\right.
     \end{align*}
     Then build $\mathcal C = (\Sigma, Q_1 \times \ldots \times Q_k, \mu, s_0, E)$
     with
     \begin{align*}
       s_0                               & = (s_0^{(1)}, \ldots, s_0^{(k)}), \\
       \mu( (t_1, \ldots, t_k), a_j ) & = (t_1, \ldots, t_{j-1}, \delta_j(t_j, a_j), t_{j+1}, \ldots, t_k) \quad \mbox{ for all } 1 \le j \le k, \\
       E & = \{ \mu( s_0, u) : \sigma(\psi(u)) \in \mathcal F \}.
     \end{align*}
     By construction, for words $u,v \in \Sigma$ with $u \in \perm(v)$
     we have $\mu((t_1, \ldots, t_k), u) = \mu((t_1, \ldots, t_k), v)$
     for any state $(t_1, \ldots, t_k) \in Q_1 \times \ldots Q_k$.
     Hence, the language accepted by $\mathcal C$ is commutative.
     We will show that $L(\mathcal C) = \{ u \in \Sigma^* \mid \sigma(\psi(u)) \in \mathcal F \}$.
     By choice of $E$ we have $\{ u \in \Sigma^* \mid \sigma(\psi(u)) \in \mathcal F \}
     \subseteq L(\mathcal C)$.
     Conversely, suppose $w \in L(\mathcal C)$.
     Then $\mu(s_0, w) = \mu(s_0, u)$ for some $u \in \Sigma^*$
     with $\sigma(\psi(u)) \in \mathcal F$.
     Next, we will argue that we can find $w' \in L(\mathcal C)$
     and $u' \in \Sigma^*$ with $\sigma(\psi(u')) \in \mathcal F$, $\mu(s_0, w') = \mu(s_0, w) = \mu(s_0, u) = \mu(s_0, u')$ and $\max\{ |w'|_{a_j}, |u'|_{a_j} \} < I_j + P_j$
      for all $j \in \{1,\ldots, k\}$.
     
     \begin{enumerate}
     \item[(i)] By construction of $\mathcal C$, if $|w|_{a_j} \ge I_j + P_j$, we
      can find $w'$ with $|w'|_{a_j} = |w|_{a_j} - P_j$ such that
      $\mu(s_0, w') = \mu(s_0, w)$. So, applying this procedure repeatedly, we can find $w' \in \Sigma^*$
      with $|w'|_{a_j} < I_j + P_j$ for all $j \in \{1,\ldots, k\}$
      and $\mu(s_0, w) = \mu(s_0, w')$.

     \item[(ii)] If $|u|_{a_j} \ge I_j + P_j$, by Equation \eqref{eqn:state_label_fn_ult_periodic},
      we can find $u'$ with $|u'|_{a_j} = |u|_{a_j} - P_j$
      and $\sigma(\psi(u'))\in \mathcal F$.
    %   By Lemma~\ref{lem:state_label_ind_form} we have
    %   $$
    %   \delta(p) = \bigcup_{p = \psi(w)} f(\sigma(0,\ldots, 0), w).
    %   $$
    %  
    %   By definition of the state label function, as some word must induce the final state in $\sigma_{\mathcal A}(\psi(u'))$,
    %   we can choose $u' \in L(\mathcal A)$.
      By construction of $\mathcal C$, we have $\mu(s_0, u) = \mu(s_0, u')$.
      So, after repeatedly applying the above steps, we find $u' \in \Sigma^*$
      with $\sigma(\psi(u')) \in \mathcal F$, $\mu(s_0,u) = \mu(s_0, u')$ and $|u'|_{a_j} < I_j + P_j$ for all $j \in \{1,\ldots, k\}$.
     \end{enumerate}
     By construction of $\mathcal C$, for words $u, v \in \Sigma^*$
     with $\max\{|u|_{a_j}, |v|_{a_j}\} < I_j + P_j$
     for all $j \in \{1,\ldots, k\}$, we have 
     \begin{equation}\label{eqn:C} 
      \mu(s_0, u) = \mu(s_0, v) \Leftrightarrow u \in \perm(v) \Leftrightarrow  \psi(u) = \psi(v).
     \end{equation}
     Hence, using Equation~\eqref{eqn:C} for the words $w'$ and $u'$ from (i) and (ii) above, 
     as $\mu(s_0, u') = \mu(s_0, w')$, we
     find $\psi(u') = \psi(w')$.
     So $\sigma(\psi(w')) = \sigma(\psi(u')) \in \mathcal F$.
     Now, again using Equation~\eqref{eqn:state_label_fn_ult_periodic},
     this gives $\sigma(\psi(w)) \in \mathcal F$. $\qed$
    \end{proof}

    I refer to~\cite{DBLP:conf/dcfs/Hoffmann20} for examples and more explanations.
    % where
    %state set labels do not fulfill the condition
    %of Theorem~\ref{thm:regularity_condition} and where for some family of state sets
    %the inverse images are not regular.
   
   \subsection{Automata Induced State Label Maps}
   \label{sec:automata_induced_state_label_maps}
   
    We call a state label map $\sigma : \mathbb N_0^k \to \mathcal P(Q)$ %entscheidungsproblem?
    given by a function $f : \mathcal P(Q) \times \Sigma \to \mathcal P(Q)$
    an \emph{automaton induced state label map}, if there exists some 
    semi-automaton $\mathcal A = (\Sigma, Q, \delta)$ such that
    $\delta(S, a) \subseteq f(S, a)$ for each $a \in \Sigma$. We also say
    that such an (semi-)automaton\footnote{For every automaton $\mathcal A = (\Sigma, Q, \delta, s_0, F)$
we can consider the corresponding semi-automaton $\mathcal A = (\Sigma, Q, \delta)$
and we will do so without special mentioning.} 
    is \emph{compatible with the state label map}. 
    This gives inductively that $\delta(S, w) \subseteq f(S, w)$
    for each word $w \in \Sigma^*$ and set $S \subseteq Q$.
    
    %   The next results says that for a word and a state label at some points,
%   if we apply the transition function to the state label and that word, the resulting set of states 
%   is contained in the state label of the point resulting if we translate the point in question
%   by the Parikh image of the word.
    
    \begin{lemma} \label{lem:states_along_words}
     Let $\mathcal A = (\Sigma, Q, \delta)$ be a semi-automaton
     and suppose the state label map $\sigma : \mathbb N_0^k \to \mathcal P(Q)$
     is compatible with $\mathcal A$.
     Let $p,q \in \mathbb N_0^k$ with $q < p$, then
     $
      \delta(\sigma(q), w) \subseteq \sigma(p)
     $
     for each $w \in \Sigma^*$ with $p = \psi(w) + q$.
    \end{lemma}
    \begin{proof} Let $w \in \Sigma^*$ with $p = \psi(w) + q$.
     Set $n = |w|$. As $q < p$ we have $1 \le n \le p_1 + \ldots + p_k$. 
     Hence, by Lemma~\ref{lem:state_label_ind_form},
     $f(\sigma(q), w) \subseteq \sigma(p)$.
     As the state label map is compatible with $\mathcal A$,
     we have $\delta(\sigma(q), w) \subseteq f(\sigma(q), w)$. \qed
     \end{proof}

    Our most important result, which generalizes a corresponding
    result from~\cite{DBLP:conf/dcfs/Hoffmann20}
    to automata induced state label maps, is stated next.
    
    \begin{proposition}
    \label{prop:state_label_for_perm_aut}
     Let $\mathcal A = (\Sigma, Q, \delta, s_0, F)$
     be a permutation automaton and $\sigma : \mathbb N_0^k \to \mathcal P(Q)$
     a state map compatible with $\mathcal A$.
     Then for every automaton $\mathcal A_p^{(j)}$
     from Definition~\ref{def:sequ_grid_decomp_aut}
     its index equals at most $(|Q|-1) L_j$
     and its period is divided by $L_j$,
     where $L_j$ denotes the order of the letter $a_j$ viewed
     as a permutation of $Q$, i.e., $\delta(q, a_j^{L_j}) = q$
     for any $q \in Q$ and $L_j$ is minimal with this property.
    \end{proposition}
    \begin{proof} %todo L_j und so festlegen
It might be helpful for the reader
to have some idea of how the symmetric group (or any permutation group) acts on subsets
of its permutation domain, %the
%first chapters of~\cite{cameron_1999} give some results and definitions in that direction.
see for example~\cite{cameron_1999} for further information. 
We also say that the letter $a_j$ acts (or operates) on a subset $S \subseteq Q$, 
the action being given by the transition function $\delta : Q \times \Sigma \to Q$,
where $\delta(S, a_j)$ is the result of the action of $a_j$ on $S$.
Set
$$
 \mathcal P = \{ \mathcal A_q^{(j)} \mid p = q + \psi(b) \mbox{ for some } b \in \Sigma \}.
$$
Denote by $I$ the maximal index and by $P$ the least common multiple
of the periods of the unary automata in $\mathcal P$.

First the case $\mathcal P = \emptyset$, which
is equivalent with $p = (0,\ldots, 0)$.
In this case, $I = 0, P = 1$, $Q_p^{(j)} = \mathcal P(Q) \times \{0\}$
and Equation~\eqref{eqn:def_unary_aut_transition}
reduces to
$$
 \delta_p^{(j)}((S, 0), a_j) = (f(S, a_j), 0)
$$
for $S \subseteq Q$. As the state label map is compatible with $\mathcal A$,
we have $\delta(S, a_j) \subseteq f(S, a_j)$.
So, as $a_j$ permutes the states $Q$, 
if $|S| = |f(S, a_j^n)|$ for $n \ge 0$, then $f(S, a_j^n) = \delta(S, a_j^n)$.
As for each $S \subseteq Q$ we have $\delta(S, a_j^{L_j}) = S$,
if $|f(S, a_j^n)| = |S|$, which gives $f(S, a_j^n) = \delta(S, a_j^n)$, we find $0 \le m < L_j$
with $f(S, a_j^m) = f(S, a_j^n)$.
Let
$$
 R = \{ n > 0 : |f(\sigma(p), a_j^{n-1})| < |f(\sigma(p), a_j^n)| \}.
$$
If $R = \emptyset$, then $f$ does not add any states as symbols are read,
and the automaton $\mathcal A_p^{(j)}$
is essentially the action of $a_j$ starting on the set $\sigma(p)$, i.e.,
the orbit $\{ \sigma(p), \delta(\sigma(p), a_j), \delta(\sigma(p), a_j^2), \ldots \}$.
Hence we have index zero and some period dividing $L_j$, as the letter $a_j$ is a permutations
of order $L_j$ on $Q$.
If $R \ne \emptyset$, then $R$ is finite, as the sets could
not grow indefinitely. Let $m = |R|$
and write $R = \{ n_i \mid i \in \{1,\ldots, m\} \}$
with $n_i < n_{i+1}$ for $i \in \{1,\ldots, m-1\}$, i.e., the sequence
orders the elements from $R$. 
We have $n_{i+1} - n_i \le L_j$ and $n_1 \le L_j$,
for if $n_i \le k < n_{i+1}$ (or $k < n_1$), then
with $S = f(\sigma(p), a_j^{n_i})$ (or $S = \sigma(p))$), as argued previously,
we find $f(S, a_j^k) = \delta(S, a_j^k)$. Assuming $n_{i+1} - n_i > L_j$ (or similarly $n_1 > L_j$)
would then yield $f(S, a_j^{L_j}) = S$, and so for every $k \ge n_i$, writing $k = qL_j + r$,
we have $f(S, a_j^k) = f(S, a_j^r)$ and the cardinalities could not grow anymore, i.e., we would be stuck in a cycle.
%By the previous considerations, 
%we have $n_{i+1} - n_i \le L_j$ and $n_1 \le L_j$, for otherwise we would have ended up in a cycle
%of the permutation action of $a_j$ on a subset of $Q$ and the sets $f(\sigma(p), a_j^n) = \delta(\sigma(p), a_j^n)$
%with $n > n_i$ cannot grow in size.
% vorsicht, f(\emptyset, a) != \emptyset möglich
%If $\sigma(p) = \emptyset$, then the statement is obviously true,
%so suppose $\sigma(p) \ne \emptyset$.
So by definition of $R$, $|\sigma(p)| < |f(\sigma(p), a_j^{n_1}| < \ldots < |f(\sigma(p), a_j^{n_m}| \le |Q|$.
This gives $m \le |Q| - |\sigma(p)|$.
By choice, for $n \ge n_m$
we have $|f(\sigma(p), a_j^{n_m})| = |f(\sigma(p), a_j^n)|$.
Hence it is again just the action of $a_j$ starting
on the subset $f(\sigma(p), a_j^{n_m})$. So we are in the cycle, and
the period of $\mathcal A_p^{(j)}$ divides $L_j$, as the operation of $\mathcal A_p^{(j)}$
could be identified with the function $f : \mathcal P(Q) \times \Sigma \to \mathcal P(Q)$ for $p = (0,\ldots, 0)$.
Note that $n_m$ is precisely the index of $\mathcal A_p^{(j)}$,
and by the previous considerations $n_m \le (|Q| - |\sigma(p)|)L_j$.

\medskip

So, now suppose $\mathcal P \ne \emptyset$.
We split the proof into several steps. 
Note that the statements (ii), (iii), (iv), (v) written below are also proven
by the above considerations for the case $\mathcal P = \emptyset$.
Hence, we can argue inductively in their proofs.
Let $S, T \subseteq Q$.

\begin{enumerate}
\item[(i)] \underline{Claim:} If $(T, y) = \delta_p^{(j)}((S, x), a_j^r)$
 for some $r \ge 0$, then $|T| \ge |S|$. In particular,
 the state labels of cycle states all have the same cardinality.
 
  \smallskip 
  
 \underline{Proof of Claim (i):} By Equation~\eqref{eqn:def_seq_grid_aut_transition}, 
 $f(S, a_j) \subseteq \pi_1(\delta_p^{(j)}((S, x), a_j))$.
 As the state label map is compatible with $\mathcal A$,
 we have $\delta(S, a_j) \subseteq f(S, a_j)$,
 and as $a_j$ is a permutation on the states, we have
 $|S| = |\delta(S, a_j)|$. 
 Hence $|S| \le |\pi_1(\delta_p^{(j)}((S, x), a_j))|$,
 which gives the claim inductively. 
 As states on the cycle could be mapped to each other, the
 state labels from cycle states all have the same cardinality.
 
 \medskip

\item[(ii)]  \underline{Claim:}  Let $L_S = \lcm\{ |\{ \delta(s, a_j^i) : i \ge 0 \}| : s \in S \}$, i.e.
 the least common multiple of the orbit lengths\footnote{For a permutation $\pi : [n] \to [n]$ 
 on a finite set $[n]$ and $m \in [n]$, the orbit length of $m$
 under the permutation $\pi$ is $|\{ \pi^i(m) : i \ge 0 \}|$. In~\cite{DBLP:conf/dcfs/Hoffmann20},
 the orbit length of an element is also called the cycle length of that element, as 
 it is precisely the size of the unique cycle in which the element $m$ appears with respect
 to the permutation.}
 of all elements in $S$.
 For $x \ge I$ and $(T,y) = \delta_p^{(j)}( (S, x), a_j^{\lcm(P, L_S)} )$,
 if $|T| = |S|$, then $(T,y) = (S,x)$. So, by Lemma~\ref{lem:unary_period_divides},
 the period of $\mathcal A_p^{(j)}$ divides $\lcm(P, L_S)$.

   \medskip

 \underline{Proof of Claim (ii):} From Equation \eqref{eqn:def_seq_grid_aut_transition} of Definition \ref{def:sequ_grid_decomp_aut} and the fact that the state map is compatible with $\mathcal A$,
       we get inductively $\delta(S, a_j^i) \subseteq f(S, a_j^i) \subseteq \pi_1(\delta_p^{(j)}((S,x), a_j^i))$
       for all $i \ge 0$.
       So, as $\delta(s, a_j^{L_S}) = s$
       for all $s \in S$,
       this gives $S \subseteq T$. Hence, as $|S| = |T|$, we get $S = T$.
       Furthermore, as $x \ge I_i$, by Equation \eqref{eqn:def_unary_aut_transition} of 
       Definition \ref{def:sequ_grid_decomp_aut}, as $P$ divides $\lcm(P, L_S)$, we have $x = y$. 
       By Lemma \ref{lem:unary_period_divides}, this implies
       that the period of $\mathcal A_p^{(j)}$ divides $\lcm(P, L_S)$.

 \medskip

\item[(iii)]  \underline{Claim:} With the notation from (ii), the number
 $\lcm(P, L_S)$ divides $L_j$
 and the  period of $\mathcal A_p^{(j)}$
 divides $L_j$.
 
 \medskip 
   
 \underline{Proof of Claim (iii):} With the notation from (ii),
 as $L_j = \lcm\{ |\delta(q, a_j^i) : i \ge 0 \}| : q \in Q \}$,
 $L_S$ divides $L_j$.
 Inductively, the periods of all unary automata in $\mathcal P$
 divide $L_j$. So, as $P$ is the least common multiple of these periods,
 also  $P$ divides $L_j$.
 Hence $\lcm(P, L_S)$ divides $L_j$. So, with Claim (ii),
 the period of $\mathcal A_p^{(j)}$ divides $L_j$.

 \medskip 
 
\item[(iv)] \underline{Claim:} For $x \ge I$
 and $(T,y) = \delta_p^{(j)}( (S, x), a_j^{L_j} )$,
 if $|T| = |S|$, then $(T,y) = (S,x)$.
 
 \medskip 
 
 \underline{Proof of Claim (iv):}
 With the notation from (ii) and Claim (iii),
 we can write $L_j = m \cdot \lcm(P, L_j)$ for some natural number $m \ge 1$.
 Set $(R, z) = \delta_p^{(j)}( (S, x), a_j^{\lcm(P, L_S)} )$.
 By (i), we have $|S| \le |R| \le |T|$.
 By assumption $|S| = |T|$, hence $|S| = |R|$.
 So, we can apply (ii), which yields $(R,z) = (S,x)$.
 Applying this repeatedly $m$ times gives $(T, y) = (S,x)$.
 
  \medskip 
  
\item[(v)] \underline{Claim:} If $T$ is the state label of any cycle state
 of $\mathcal A_p^{(j)}$, then the index of $\mathcal A_p^{(j)}$
 is bounded by $(|T|-1) L_j$.
 
   \medskip 
 
 \underline{Proof of Claim (v):}  We define a sequence $(T_n, y_n) \in Q_p^{(j)}$ of states for $n \in \mathbb N_0$.
         Set $(T_0, y_0) = \delta_p^{(j)}(s_p^{(0,j)}, a_j^I)$, which implies $y_0 = I$ by 
         Equation~\eqref{eqn:def_unary_aut_transition}, and
         $$
          (T_n, y_n) = \delta_p^{(j)}( (T_{n-1}, y_{n-1}), a_j^{L_j} ) 
         $$
         for $n > 0$. Note that, as $P$ divides $L_j$, by Equation~\eqref{eqn:def_unary_aut_transition}, 
         we have $y_n = I$ for all $n \ge 0$.
         
         \medskip 
         
         \noindent\underline{Claim 1:} Let $(T, x)\in Q_p^{(j)}$ be some state from the cycle of $\mathcal A_p^{(j)}$.
           Then the state $(T_{|T|-|T_0|}, y_{|T|-|T_0|}) = \delta_p^{(j)}( s_p^{(0,j)}, a_j^{I + (|T| - |T_0|) L_j} )$
           is also from the cycle of $\mathcal A_p^{(j)}$.
           
           \medskip
           
         By construction, and Equation \eqref{eqn:def_unary_aut_transition} 
         from Definition \ref{def:sequ_grid_decomp_aut}, we have $y_n \ge I$ for all $n$.
         If $T_{n+1} \ne T_n$, then, by (iv) and (i), we have $|T_{n+1}| > |T_n|$ (remember $y_n = y_{n+1} = I$).
         Hence\footnote{Also, as $T_n = \delta(T_n, a_J^{L_j}) \subseteq f(T_n, a_j^{L_j})
         \subseteq \pi_1(\delta_p^{(j)}((T_n, I), a_j^{L_j}))$, we find $T_n \subseteq T_{n+1}$.}, 
         by finiteness, we must have a smallest $m$ such that $T_{m+1} = T_m$.
         %By (iv), this also implies $y_{m+1} = y_m$.
         As also $y_{m+1} = y_m$, we are on the cycle of $\mathcal A_p^{(j)}$, and the period
         of this automaton divides $L_j$ by (iv).
         This yields $(T_n, y_n) = (T_m, y_m)$ for all $n \ge m$.
         By (i),
         the size of the state label sets on the cycle stays constant, and just
         grows before we enter the cycle.
         As we could
         add at most $|T_m| - |T_0|$ elements, and for $T_0, T_1, \ldots, T_{m}$
         each time at least one element is added, we have, as $m$
         was chosen minimal, that $m \le |T| - |T_0|$,
         where $T$ is any state label on the cycle, which all have the same cardinality $|T| = |T_m|$
         by (i).
         This means we could read at most $|T| - |T_0|$
         times the sequence $a_j^{L_j}$, starting from $(T_0, I)$,
         before we enter the cycle of $\mathcal A_p^{(i)}$.
         
         \medskip 
         
         \noindent\underline{Claim 2:} We have $I \le (|T_0|-1) L_j$.
          
          \medskip
          
          Remember, the case $\mathcal P = \emptyset$ was already handled, for then $p = (0,\ldots, 0)$
          and $I = 0$. Otherwise, let $\mathcal A_q^{(j)} \in \mathcal P$ with $p = q + \psi(b)$
          for $b \in \Sigma$. 
          Let $(S, x) = \delta_q^{(j)}( s_q^{(0,j)}, a_j^n )$ with $n \ge 0$.
          If $n > 0$, by Equation \eqref{eqn:def_seq_grid_aut_transition}
          from Definition \ref{def:sequ_grid_decomp_aut}, 
          we have, with $(R, z) = \delta_p^{(j)}(s_p^{(0,j)}, a_j^{n-1})$,
          \begin{align*} 
            \pi_1(\delta_p^{(j)}(s_p^{(0,j)}, a_j^n)) 
             & = \pi_1(\delta_p^{(j)}((R,z), a_j)) \\
             & = f(R, a_j) \cup \bigcup_{\substack{(r, a) \in \mathbb N_0^k \times\Sigma\\ p = r + \psi(a)}} f(\pi_1(\delta_r^{(j)}(s_r^{(0,j)}, a_j^n)), a).
          \end{align*}
          If $n = 0$, we have
          $$
            \pi_1(\delta_p^{(j)}(s_p^{(0,j)}, a_j^n)) = \pi_1((\sigma(p), 0)) = \sigma(p).
          $$
          In the latter case, also
          $S = \pi_1(\delta_q^{(j)}( s_q^{(0,j)}, a_j^0) ) = \sigma(q)$
          and, as $p \ne (0, \ldots 0)$ (which is equivalent to $\mathcal P \ne \emptyset$),
          by Equation~\eqref{eqn:state_label_fn} and
          as the state label map is compatible with $\mathcal A$, we have
          $\delta(S, b) \subseteq f(S,b) \subseteq \sigma(p)$.
          In the former case $n > 0$,
          $$
           \delta(S, b) \subseteq f(S,b) \subseteq \bigcup_{\substack{(r, a) \in \mathbb N_0^k \times\Sigma\\ p = r + \psi(a)}} f(\pi_1(\delta_r^{(j)}(s_r^{(0,j)}, a_j^n)), a)
          $$
          So, in any case, $\delta(S, b) \subseteq \pi_1(\delta_p^{(j)}(s_p^{(0,j)}, a_j^n))$.
          In particular for $n = I$ we get $\delta(S, b) \subseteq T_0$, and
          as $b$ induces a permutation on the states, this gives $|S| \le |T_0|$.
          Also for $n \ge I$, we are on the cycle of $\mathcal A_q^{(j)}$.
          Hence, inductively, the index of $\mathcal A_q^{(j)}$ is at most $(|S|-1) L_j \le (|T_0|-1) L_j$.
          As $\mathcal A_q^{(j)} \in \mathcal P$ was chosen arbitrary,
          we get $I \le (|T_0|-1)L_j$.
       
         \medskip
         % todo, alles besser einrücken, dass man erkennt was zu claim 1 und 2 gehört
         \bigskip 
         
         With Claim (2) above, we can derive the upper bound $(|T|-1) L_j$ for the length of the word $ a_j^{I + (|T| - |T_0|)}$ from Claim (1),
         as 
         $$
          I + (|T| - |T_0|)L_j \le (|T_0|-1) L_j + (|T| - |T_0|) L_j = (|T|-1) L_j.
         $$
         And as Claim (1) essentially says that the index of $\mathcal A_p^{(j)}$
         is smaller than $I + (|T| + |T_0|) L_j$, this gives Claim (v). 
         Also, as $|T| \le |Q|$, the claim about the index %todo komma nach |Q|?
         of the statement in Proposition~\ref{prop:state_label_for_perm_aut} 
         is proven. So, in total, (iii) and (v) give Proposition~\ref{prop:state_label_for_perm_aut}. $\qed$
\end{enumerate}
\end{proof}

%% file: proof_it_shuffle_prelim.tex
First, we need
   to define our state label map.
   
    \begin{definition}\label{def:state_label_it_shuffle}
      Let $\mathcal A = (\Sigma, Q, \delta, s_0, F)$
      be a finite automaton. Denote by $\sigma_{\mathcal A, +} : \mathbb N_0^k \to \mathcal P(Q)$
      the state label function given by
      $f : P(Q) \times \Sigma \to P(Q)$,
      where
      \begin{equation}\label{eqn:def-state_label_it_shuffle}
       f(Q, a) = \left\{ 
       \begin{array}{ll}
        \delta(Q, a) \cup \{ s_0 \} & \mbox{if } \delta(Q, a) \cap F \ne \emptyset; \\ 
        \delta(Q, a)                & \mbox{otherwise;}
       \end{array}\right.
      \end{equation}
      and $\sigma_{\mathcal A, +}(0, \ldots, 0) = \{ s_0 \}$.
   \end{definition}

   To derive our results, we need the following formula for the image of the state label map
   at a given point.
   
\begin{proposition}
    \label{prop:it_shuffle_state_label_map_new}
    Let $\Sigma = \{a_1, \ldots, a_k\}$ and $\mathcal A = (\Sigma, Q, \delta, s_0, F)$
    be a finite automaton. For the state-label function
    from Definition~\ref{def:state_label_it_shuffle}
    we have
    % \begin{align} 
    %  \sigma_{\mathcal A, *}(p) & = \{ \delta(s_0, w) \mid \psi(w) = p \} \label{eqn:state_label_it_set} \\
    %                           & \cup \{ \delta(s_0, w) \mid \exists q \in \mathbb N_0^k : \psi(w) + q = p \mbox{ and } \sigma_{\mathcal A, *}(q) \cap F \ne \emptyset \}. \nonumber
    % \end{align}
    $$
     \sigma_{\mathcal A, +}(p) = \left\{ 
     \begin{array}{ll}
      A_p \cup B_p & \mbox{if } (A_p \cup B_p) \cap F = \emptyset; \\ 
      A_p \cup B_p \cup \{s_0\} & \mbox{otherwise;}
     \end{array}
     \right.
    $$
    where $A_p = \{ \delta(s_0, w) \mid \psi(w) = p \}$
    and $B_p = \{ \delta(s_0, w) \mid \exists q \in \mathbb N_0^k : q < p \mbox{ and } q + \psi(w) = p \mbox{ and } \sigma_{\mathcal A, +}(q) \cap F \ne \emptyset \}$.
\end{proposition}
\begin{proof}
For $p = (0,\ldots, 0)$ the statement is clear. If $p \ne (0,\ldots,0)$, then by definition
\begin{equation}\label{eqn:sigma_def}
 \sigma_{\mathcal A,+}(p) = \bigcup_{\substack{(q,b) \\ p = q + \psi(b)}} f(\sigma_{\mathcal A,+}(q), b). 
\end{equation}
For $q$ with $p = q + \psi(b)$ for some $b \in \Sigma$ set
\begin{align*}
 A_q = & \{ \delta(s_0, w) \mid \psi(w) = q \} \\
 B_q = & \{ \delta(s_0, w) \mid \exists r \in \mathbb N_0^k : r < q \mbox{ and } r + \psi(w) = q \mbox{ and } \sigma_{\mathcal A,+}(r) \cap F \ne \emptyset \}.
\end{align*}
Inductively, 
$$
 \sigma_{\mathcal A,+}(q) = \left\{ 
 \begin{array}{ll}
  A_q \cup B_q & \mbox{if } (A_q \cap B_q) \cap F = \emptyset, \\ 
  A_q \cup B_q \cup \{s_0\} & \mbox{otherwise.}
 \end{array}
 \right.
$$
Hence, using Definition~\ref{def:state_label_it_shuffle}, $f(\sigma_{\mathcal A,+}(q), b)$ equals
\begin{equation}\label{eqn:f_it_shuffle}
 \left\{ 
 \begin{array}{ll}
  \delta(A_q \cup B_q, b)              & \mbox{if } (A_q \cup B_q) \cap F = \emptyset, \delta(A_q \cup B_q, b) \cap F = \emptyset, \\ 
  \delta(A_q \cup B_q, b) \cup \{s_0\} & \mbox{if } (A_q \cup B_q) \cap F = \emptyset,  \delta(A_q \cup B_q, b) \cap F \ne \emptyset \\ 
  \delta(A_q \cup B_q \cup \{s_0\}, b)              & \mbox{if } (A_q \cup B_q) \cap F \ne \emptyset, \delta(A_q \cup B_q \cup \{s_0\}, b) \cap F = \emptyset, \\ 
  \delta(A_q \cup B_q \cup \{s_0\}, b) \cup \{s_0\} & \mbox{if } (A_q \cup B_q) \cap F \ne \emptyset,  \delta(A_q \cup B_q \cup \{s_0\}, b) \cap F \ne \emptyset. \\
 \end{array}
 \right.
\end{equation}

Under the induction hypothesis, i.e., that the formula holds true for $q < p$, in particular if $p = q + \psi(b)$
for some $b \in \Sigma$, we prove various claims that we use to derive our final formula.

\medskip 

\noindent\underline{Claim 1:} For $q \in \mathbb N_0^k$ with $p = q + \psi(b)$ for some $b \in \Sigma$ we have
$$
 (A_q \cup B_q) \cap F \ne \emptyset \Leftrightarrow \sigma_{\mathcal A,+}(q) \cap F \ne \emptyset.
$$
\begin{quote}
 \textit{Proof of the claim.} If $(A_q \cup B_q) \cap F \ne \emptyset$, then $\sigma_{\mathcal A, *}(q) \cap F \ne \emptyset$
 by induction hypothesis.
 If $\sigma_{\mathcal A, *}(q) \cap F \ne \emptyset$, assume $(A_q \cup B_q) \cap F = \emptyset$.
Then, using inductively that the formula holds true for $q$, this gives
$\sigma_{\mathcal A,+}(q) = A_q \cup B_q$, which implies $(A_q \cup B_q) \cap F \ne \emptyset$.
Hence, this is not possible and we must have $(A_q \cup B_q) \cap F \ne \emptyset$. \qed
\end{quote}

\medskip

\noindent\underline{Claim 2:} We have
\begin{align*}
 A_p & = \bigcup_{\substack{(q,b) \\ p = q + \psi(b)}} \delta(A_q, b), \\ 
 B_p & = \bigcup_{\substack{(q,b) \\ p = q + \psi(b)}} \delta(B_q, b) \cup \bigcup_{\substack{(q,b) \\ p = q + \psi(b) \\ \sigma_{\mathcal A,+}(q) \cap F \ne \emptyset}} \delta(\{s_0\}, b).
\end{align*}
\begin{quote}
 \textit{Proof of the claim.} The first equation is obvious. 
 For the other, first let $\delta(s_0, w) \in B_p$ for some $w \in \Sigma^*$.
 Then, we have $r \in \mathbb N_0^k$ such that
 $$
  r < p, \quad r + \psi(w) = p, \mbox{ and } \sigma_{\mathcal A,+}(r) \cap F \ne \emptyset.
 $$
 Write $w = ub$ with $b \in \Sigma$ (note that by definition of the sets $B_p$
 we have $|w| > 0$ here). If $r < p - \psi(b)$,
 then $\delta(s_0, u) \in B_{p - \psi(b)}$
 and so $\delta(s_0, w) \in \delta(B_{p - \psi(b)}, b)$.
 Otherwise $r = p - \psi(b)$, which implies $u = \varepsilon$ and $w = b$. 
 In this case, 
 $$
  \delta(s_0, b) \in \bigcup_{\substack{(q,b) \\ p = q + \psi(b) \\ \sigma_{\mathcal A,+}(q) \cap F \ne \emptyset}} \delta(\{s_0\}, b).
 $$
 Hence, $B_p$ is included in the set on the right hand side. The inclusion of the other two sets
 in $B_p$ is obvious. $\qed$
\begin{comment}
 For the other, note
 that
 $$
  \bigcup_{\substack{(q,b) \\ p = q + \psi(b)}} \delta(B_q, b)
 $$
 equals
 $$
  \{ \delta(s_0, w) \mid \exists r \in \mathbb N_0^k \forall b \in \Sigma : r < p - \psi(b), \psi(w) + r = p, \sigma_{\mathcal A,+}(r) \cap F \ne \emptyset \}.
 $$
 For if $s \in \bigcup_{\substack{(q,b) \\ p = q + \psi(b)}} \delta(B_q, b)$, then
 there exists $q \in \mathbb N_0^k$ and $b \in \Sigma$
 such that $s \in \delta(B_q, b)$ and $p = q + \psi(b)$.
 Let $t \in B_q$ with $s = \delta(t, b)$.
 Then, $t = \delta(s_0, w)$ for some $w \in \Sigma^*$ and $r \in \mathbb N_0^k$
 such that $r < q$, $q = r + \psi(w)$ and $\sigma_{\mathcal A,+}(r) \cap F \ne \emptyset$.
 Then $s = \delta(s_0, wb)$ with 
 $$
  r < p - \psi(b), p = r + \psi(wb), \sigma_{\mathcal A,+}(r) \cap F \ne \emptyset.
 $$
 Conversely, assume $  = \delta(s_0, w)$ such that there exists $r \in \mathbb N_0^k$
 and $b \in \Sigma$ such that
 $$
   r < p - \psi(b), p = r + \psi(w), \sigma_{\mathcal A,+}(r) \cap F \ne \emptyset.
 $$

 By Claim (1), this gives that
 $$
  \left( \bigcup_{\substack{(q,b) \\ p = q + \psi(b)}} \delta(B_q, b) \right) \cup \left( \bigcup_{\substack{(q,b) \\ p = q + \psi(b) \\ (A_q \cap B_q) \cap F \ne \emptyset}} \delta(\{s_0\}, b) \right)
 $$
 equals
 $
 \{ \delta(s_0, w) \mid \exists r \in \mathbb N_0^k : r < p, \psi(w) + r = p, \sigma_{\mathcal A,+}(r) \cap F \ne \emptyset\},
$
 i.e., the set $B_p$. \qed
\end{comment}
\end{quote}

\medskip 

\noindent\underline{Claim 3:} We have $(A_p \cup B_p) \cap F \ne \emptyset$ if and only if
there exists $q \in \mathbb N_0^k$ and $b \in \Sigma^*$ with $p = q + \psi(b)$ such that 
at least one of the conditions is fulfilled:
\begin{enumerate}
\item[(1)] $(A_q \cup B_q) \cap F = \emptyset$ and $\delta(A_q \cup B_q, b) \cap F \ne \emptyset$,
\item[(2)] $(A_q \cup B_q) \cap F \ne \emptyset$ and $\delta(A_q \cup B_q \cup \{s_0\}, b) \cap F \ne \emptyset$.  
\end{enumerate}
\begin{quote}
 \textit{Proof of the claim.} Assume $(A_p \cup B_p) \cap F \ne \emptyset$.
 We distinguish the two cases $A_p \cap F \ne \emptyset$ or $B_p \cap F \ne \emptyset$.
 First, suppose $A_p \cap F \ne \emptyset$.
 By Claim (2) then $\delta(A_q, b) \cap F \ne \emptyset$ 
 for some $q \in \mathbb N_0^k$ and $b \in \Sigma$ with $p = q + \psi(b)$.
 As $\delta(A_q, b) \subseteq \delta(A_q \cup B_q, b) \subseteq \delta(A_q \cup B_q \cup \{s_0\}, b)$,
 both conditions (1) and (2) are fulfilled.
 Now, suppose $B_p \cap F \ne \emptyset$.
 Using Claim (2), we have two cases.
 \begin{enumerate}
 \item It is $\delta(B_q, b) \cap F \ne \emptyset$ for some $q \in \mathbb N_0$ and $b \in \Sigma$ with $p = q + \psi(b)$.
 
    As $\delta(B_q, b) \subseteq \delta(A_q \cup B_q, b) \subseteq \delta(A_q \cup B_q \cup \{s_0\}, b)$,
    both conditions (1) and (2) are fulfilled.
    
 \item We find, also using Claim (1), some $q \in \mathbb N_0^k$
 and $b \in \Sigma$ with $p = q + \psi(b)$ and
 $(A_q \cup B_q) \cap F \ne \emptyset$ such that $\delta(s_0, b) \in F$.     
 
   Then condition (2) is fulfilled.
  
 \end{enumerate}
 
 Conversely, assume condition (1) is fulfilled. Then, by Claim (2),
 we have $A_p \cap F \ne \emptyset$ or $B_p \cap F \ne \emptyset$.
 Otherwise, assume condition (2) is fulfilled.
 If $\delta(A_q \cup B_q, b) \cap F \ne \emptyset$, we have $(A_p \cup B_p) \cap F \ne \emptyset$
 as before. So, assume $\delta(A_q \cup B_q, b) \cap F = \emptyset$.
 But then, we must have $\delta(s_0, b) \in F$, using Claim (2),
 which gives, as $(A_q \cup B_q) \cap F \ne \emptyset$ and using Claim (1) and Claim (2),
 that $\delta(s_0, b) \in B_p$, hence $B_p \cap F \ne \emptyset$.~$\qed$
\end{quote}

\noindent First, assume $(A_p \cup B_p) \cap F = \emptyset$.
Then, by Equation~\eqref{eqn:f_it_shuffle} together with Claim (3) 
and Equation~\eqref{eqn:sigma_def}, 
\begin{align*} 
 \sigma_{\mathcal A,+}(p) & = \bigcup_{\substack{(q,b) \\ p = q + \psi(b)}} f(\sigma_{\mathcal A,+}(q), b) \\ 
                          & = \bigcup_{\substack{(q,b) \\ p = q + \psi(b)}} \delta(A_q \cup B_q, b) \cup \bigcup_{\substack{(q,b) \\ p = q + \psi(b) \\ (A_q \cap B_q) \cap F \ne \emptyset}} \delta(\{s_0\}, b). 
\end{align*}
By Claim (1) and Claim (2), we get
$
 \sigma_{\mathcal A,+}(p) = A_p \cup B_q.
$
Otherwise, if $(A_p \cup B_p) \cap F \ne \emptyset$, by Equation~\eqref{eqn:f_it_shuffle} together with Claim (3) 
and Equation~\eqref{eqn:sigma_def}, 
\begin{align*} 
 \sigma_{\mathcal A,+}(p) & = \bigcup_{\substack{(q,b) \\ p = q + \psi(b)}} f(\sigma_{\mathcal A,+}(q), b) \\ 
                          & = \{s_0\} \cup \bigcup_{\substack{(q,b) \\ p = q + \psi(b)}} \delta(A_q \cup B_q, b) \cup \bigcup_{\substack{(q,b) \\ p = q + \psi(b) \\ (A_q \cap B_q) \cap F \ne \emptyset}} \delta(\{s_0\}, b). 
\end{align*}
As above, this equals $\{s_0\} \cup A_p \cup B_p$. \qed

\end{proof}

   The next statement
   %, in analogy to Propostion~\ref{prop:parikh_image_concat}
   %from Section~\ref{sec:shuffle}, 
   gives a connection
   between the Parikh image of $L(\mathcal A)^*$
   and $\sigma_{\mathcal A, +} : \mathbb N_0^k \to \mathcal P(Q)$.
   
   \begin{proposition}
   \label{prop:it_shuffle_parikh_map}
    Let $\mathcal A = (\Sigma, Q, \delta, s_0, F)$
    be a finite automaton. %For the state-label function
    %from Definition~\ref{def:state_label_it_shuffle} %komma,todo?
    %we have
    Then
    $
     \psi(L(\mathcal A)^*) =  \sigma_{\mathcal A, +}^{-1}( \{ S \subseteq Q \mid S \cap F \ne \emptyset \} ) \cup \{ (0,\ldots, 0)\}.
    $
   \end{proposition}
  \begin{proof} 
%   \begin{comment}
%      Suppose we have given some $n > 1$.
%      Let $\mathcal A_i = (\Sigma, Q_i, \delta_i, s_i, F_i)$
%      for $i \in \{1,\ldots, n\}$ be disjoint copies of $\mathcal A$. % vlt genauer erklären, todo?
%      Let $\sigma_{\mathcal A_1, \ldots, \mathcal A_n} : \mathbb N_0^k \to \mathcal P(Q_1 \cup \ldots \cup Q_n)$
%      be the state label function according to Definition~\ref{todo}.
%      We write $\sigma_{\mathcal A^n}$ for short. %oder das auch in dem kapitel einführen, todo

%      \begin{enumerate}
%      \item For each $S_n \subseteq Q_n$ we have $\sigma_{\mathcal A^n}^{-1}(S_n) \subseteq \sigma_{\mathcal A, *}^{-1}(S)$.

%       If $n = 1$.  entsrpechend f einfach delta definieren. dann kommautive closure, ist def aus grp cocoon paper.
%       da f automata induced, also delta immer drin klar bzw induktionsaussage wie in coocon paper 
%       delte = { delta(s, w) .... }
       
%      \item     
%      \end{enumerate}
%   \end{comment}
   First suppose $p \in \psi(L(\mathcal A)^*)$. Then either $p = (0,\ldots, 0)$
   or we find $p_1, \ldots, p_n$
   with $n > 0$, $p = p_1 + \ldots + p_n$ and words $w_1, \ldots, w_n$
   with $p_i = \psi(w_i)$ and $w_i \in L(\mathcal A)$ for $i \in \{1,\ldots, n\}$.
   If $p \ne (0,\ldots, 0)$, then we can assume $w_i \ne \varepsilon$ for $i \in \{1,\ldots,n\}$, which is equivalent with $p_i \ne (0,\ldots, 0)$. 
%   Also, we choose the $w_j$ minimal in the sense that, 
%   if $w_1, \ldots, w_{j-1}$ have been choosen, then choose $w_j$ of minimal length
%   in 
%   $$
%     w_j \in \{ u \in L(\mathcal A)  \mid \exists u_1 \in L(\mathcal A) \cdots \exists u_m \in L(\mathcal A) : \psi(w_1 \cdots w_{j-1} u u_1 \cdots u_m) = p \}.
%   $$
%   This choice guarantees the next claim.
   Note that, with the notation from Proposition~\ref{prop:it_shuffle_state_label_map_new},
   for any $p \in \mathbb N_0^k$
   \begin{equation}
       \label{eqn:state_map_it_shuffle_final_states}
       \sigma_{\mathcal A,+}(p) \cap F \ne \emptyset \Leftrightarrow 
       ( A_p \cup B_p ) \cap F \ne \emptyset. 
   \end{equation}
   For if $\sigma_{\mathcal A,+}(p) \cap F = \emptyset$
   then obviously $(A_p \cup B_p) \cap F = \emptyset$.
   And if $(A_p \cup B_p) \cap F = \emptyset$ holds true, then $\sigma_{\mathcal A,+}(p) = A_p \cup B_p$.
   So, $\sigma_{\mathcal A,+}(p) \cap F \ne \emptyset$ implies $s_0 \in \sigma_{\mathcal A,+}(p)$.
   
   \medskip
   
   \noindent\underline{Claim:} For $i \in \{1,\ldots, n\}$ we have
   $\sigma_{\mathcal A,+}(p_1 + \ldots + p_i) \cap F \ne \emptyset$.
   
   \begin{quote}
       \textit{Proof of the claim.} 
       As $w_1 \in L(\mathcal A)$ and for every $w \in \Sigma^*$
       by Definition~\ref{def:state_label_it_shuffle} and Lemma~\ref{lem:state_label_ind_form}
       we have $\delta(s_0, w) \in f(\{s_0\}, w) \subseteq \sigma_{\mathcal A,+}(\psi(w))$,
       we get $\delta(s_0, w_1) \in \sigma_{\mathcal A,+}(p_1)$.
       Hence $\sigma_{\mathcal A,+}(p_1) \cap F \ne \emptyset$ as $\delta(s_0, w_1) \in F$.
       Now, suppose inductively that for $i \in \{1,\ldots, n-1\}$ we have
       $$
        \sigma_{\mathcal A,+}(p_1 + \ldots + p_i) \cap F \ne \emptyset.
       $$
       By Equation~\eqref{eqn:state_map_it_shuffle_final_states}
       and the remarks thereafter, $s_0 \in \sigma_{\mathcal A,+}(p_1 + \ldots + p_i)$.
       By Definition~\ref{def:state_label_it_shuffle} and Lemma~\ref{lem:state_label_ind_form}
       then, as $p_1 + \ldots p_{i+1} = p_1 + \ldots p_i + \psi(w_{i+1})$,
       \begin{align*} 
        \delta(s_0, w_{i+1}) & \in \delta(\sigma_{\mathcal A,+}(p_1 + \ldots + p_i), w_{i+1}) \\
                             & \subseteq f(\sigma_{\mathcal A,+}(p_1 + \ldots + p_i), w_{i+1}) & \mbox{[Definition~\ref{def:state_label_it_shuffle}]} \\
                             & \subseteq \sigma_{\mathcal A,+}(p_1 + \ldots + p_i + p_{i+1}). & \mbox{[Lemma~\ref{lem:state_label_ind_form}]}
       \end{align*}
       As $\delta(s_0, w_{i+1}) \in F$ we find $\sigma_{\mathcal A,+}(p_1 + \ldots + p_i + p_{i+1}) \cap F \ne \emptyset$.$\qed$
    %   Assume $\sigma_{\mathcal A, +}(p_1) \cap F = \emptyset$.
    %   Denote by $A_{p_1}$ and $B_{p_1}$ the sets from Proposition~\ref{prop:it_shuffle_state_label_map_new}.
    %   Then, by this proposition, we have $( A_{p_1} \cup B_{p_1} ) \cap F = \emptyset$
    %   and so $\sigma_{\mathcal A,+}(p_1) = A_{p_1} \cup B_{p_1}$.
    %
    %   We have $\delta(s_0, w_1) \in \sigma_{\mathcal A,+}(p_1)$.
    %     For, if otherwise, by Proposition~\ref{prop:it_shuffle_state_label_map_new},
    %     with the notation from the referred proposition, 
   \end{quote}
   With the above claim, for $i = n$, we find $\sigma_{\mathcal A,+}(p)\cap F \ne \emptyset$.
%   
%   Choose $p_1$ minimal but non-zero, then $p_2$ minimal but non-zero with this choice of $p_1$ and so on until $p_n$. 
%   Then, by Proposition~\ref{prop:it_shuffle_state_label_map_new}, $\delta(s_0, w_1) \in \sigma_{\mathcal A,+}(p_1)$,
%   as $\delta(s_0, w_1) \in A_p$ with the notation from Proposition~\ref{prop:it_shuffle_state_label_map_new}
%   and the minimality of $p_1$.
%   So, as $\delta(s_0, w_1) \in F$, we have $\sigma_{\mathcal A,+}(p_1) \cap F \ne \emptyset$.
%   Now, suppose inductively that for $j \in \{1,\ldots, n-1\}$ we have
%   $$
%      \sigma_{\mathcal A,+}(p_1 + \ldots + p_j) \cap F \ne \emptyset.
%   $$
%   By Proposition~\ref{prop:it_shuffle_state_label_map_new} and the minimality $p_{j+1}$, as $\psi(w_{j+1}) + p_j + \ldots + p_1 = p_{j+1} + \ldots + p_1$
%   we have
%   $$
%     \delta(s_0, w_{j+1}) \in \sigma_{\mathcal A,+}(p_1 + \ldots + p_{j+1})
%   $$
%   which yields $\sigma_{\mathcal A,+}(p_1 + \ldots + p_{j+1}) \cap F \ne \emptyset$.
%   So, for $j = n$, $\sigma_{\mathcal A,+}(p)\cap F \ne \emptyset$.

   Conversely, assume $\sigma_{\mathcal A,+}(p) \cap F \ne \emptyset$ or $p = (0,\ldots, 0)$.
   In the latter case we have $p \in \psi(L(\mathcal A)^*)$ by definition of the star operation.
   Hence, assume the former holds true.
   If $p \in \psi(L(\mathcal A)) \subseteq \psi(L(\mathcal A)^*)$ we have nothing to prove.
   So, assume $p \notin \psi(L(\mathcal A))$. Then, we claim the next.
   
   \medskip 
   
   \noindent\underline{Claim:} There exists %$q \in \mathbb N_0^k \setminus \{ (r_1, \ldots, r_k) \in \mathbb N_0^k \mid \forall j \in \{1,\ldots, k\} : r_i \ge p_i\}$ such that
   $q \in \mathbb N_0^k$ with $q < p$ such that
   $
    p = q + \psi(w)
   $
   for some $w \in L(\mathcal A)$ and $\sigma_{\mathcal A,+}(q) \cap F \ne \emptyset$.
   \begin{quote}
       \textit{Proof of the claim.} 
    %   Assume for all $q < p$ we have
    %   $\sigma_{\mathcal A, *}(q) \cap F = \emptyset$.
    %   As $p \notin \psi(L(\mathcal A))$, we have
    %   $$
    %     \{ \delta(s_0, w) \mid \psi(w) = p \} \cap F = \emptyset.
    %   $$
    %   Then by Lemma~\ref{lem:no_final_in_between}
    %   $$
    %     \sigma_{\mathcal A,+}(p) = \{ \delta(s_0, w) \mid \psi(w) = p \}.
    %   $$
    %   But again, the assumption $p \notin \psi(L(\mathcal A))$
    %   then gives $\sigma_{\mathcal A,+}(p) \cap F = \emptyset$, contradicting
    %   the assumption $\sigma_{\mathcal A,+}(p) \cap F \ne \emptyset$.
    %   So, we must have some $q < p$
    %   with $\sigma_{\mathcal A, *}(q) \cap F \ne\emptyset$.
      As $p \notin \psi(L(\mathcal A))$, we have
      $$
        \{ \delta(s_0, w) \mid \psi(w) = p \} \cap F = \emptyset.
      $$
    %   Hence, by Proposition~\ref{prop:it_shuffle_state_label_map}, we have
    %   $$
    %   \{ \delta(s_0, w) \mid \exists q \in \mathbb N_0^k : \psi(w) + q = p \mbox{ and } \sigma_{\mathcal A,+}(q) \cap F \ne \emptyset\} \cap F \ne \emptyset. 
    %   $$
    %   Choose $w \in \Sigma^*$ with $\delta(s_0, w) \in F$
    %   and $p = \psi(w) + q$ for some $q \in \mathbb N_0^k$.
     Set 
     $
      B_q = \{ \delta(s_0, w) \mid \exists q \in \mathbb N_0^k : q < p, \psi(w) + q = p, \sigma_{\mathcal A,+}(q) \cap F \ne \emptyset\}.
     $
     Assume $B_p \cap F = \emptyset$, then by Proposition~\ref{prop:it_shuffle_state_label_map_new}
     this implies $\sigma_{\mathcal A,+}(p) = \{ \delta(s_0, w) \mid \psi(w) = p \} \cup B_p$.
     But then, as $\sigma_{\mathcal A,+}(p) \cap F \ne \emptyset$, this is not possible
     and we  must have $B_p \cap F \ne \emptyset$, which gives the claim. $\qed$
   \end{quote}
   
   By the above claim, choose $q \in \mathbb N_0^k$
   with $q < p$ and $p = q + \psi(w)$ for some $w \in L(\mathcal A)$
   and $\sigma_{\mathcal A,+}(q) \cap F \ne \emptyset$.
   By induction hypothesis, we find $u \in L(\mathcal A)^*$
   with $\psi(u) = q$.
   Then $p = \psi(u) + \psi(w) = \psi(uw)$ and we have $uw \in L(\mathcal A)^*$,
   i.e. $p \in \psi(L(\mathcal A)^*)$.~\qed
   %By Proposition~\ref{prop:it_shuffle_state_label_map}, either $p \in \psi(L(\mathcal A))$ or
   %$p \notin \psi(L(\mathcal A))$.
   %but
%    we have $w \in L(\mathcal A)$ and $q \in \mathbb N_0^k$ with 
 %   $p = q + \psi(w)$ and $\sigma_{\mathcal A,+}(q) \cap F \ne \emptyset$.
  %  Note that $p \notin \psi(L(\mathcal A))$ implies $q \ne (0,\ldots,0)$.
    %If $w = \varepsilon$, then $p = q$.
    %But we also must have some $w \ne \varepsilon$ fulfilling
    %this condition by Lemma~\ref{lem:no_final_in_between},
    %as $p \notin \psi(L(\mathcal A))$.
    %For if not, then by Lemma~\ref{lem:no_final_in_between}
    %we would have $\{ \delta(s_0, w) \mid \psi(w) = p \} \cap F \ne \emptyset$,
    %but by assumption we have $p \notin \psi(L(\mathcal A))$.
   % But then $q < p$, and inductively from $\sigma_{\mathcal A, *}(q) \cap F \ne \emptyset$
%    we can conclude $q \in \psi(L(\mathcal A)^*)$.
 %   Then, as $w \in L(\mathcal A)$, we find $p \in L(\mathcal A)^*$.
   \end{proof}
   
   We will also need the next lemma.

  \begin{lemma}
\label{lem:sc_adding_empty_word}
   Let $\Sigma = \{a_1, \ldots, a_k\}$ and $L\subseteq \Sigma^*$ 
   be a regular language with $\stc(L) = n$.
   Then $\stc(L \cup \{\varepsilon\}) \le n + 1$
   and this bound is sharp. If $L$ is commutative with index vector $(i_1, \ldots, i_k)$,
   then the index vector of $L \cup \{\varepsilon\}$
   is at most $(i_1+1, \ldots, i_k+1)$
   and both languages have the same period.
\end{lemma}   
\begin{proof}
 Let $\mathcal A = (\Sigma, Q, \delta, s_0, F)$ 
 be an automaton for $L$.
 Choose $s_0' \notin Q$ and construct $\mathcal A' = (\Sigma, Q\cup \{s_0'\}, \delta', s_0', F \cup \{s_0'\})$
 with
 $$
  \delta'(s_0', x) = \delta(s_0, x)
 $$
 for $x \in \Sigma$, and $\delta'(q, x) = \delta(q, x)$ for $q \in Q$, $x \in \Sigma$.
 
 \begin{enumerate}
 \item $L(\mathcal A') \subseteq L \cup \{\varepsilon\}$.
     
    Let $w \in \Sigma^*$ be a word with $\delta'(s_0', w) \in F \cup \{s_0'\}$.
    By construction, if $\delta'(s_0', w) = s_0'$, then $w = \varepsilon$.
    Otherwise, if $\delta'(s_0', w) \ne s_0'$, then $|w| > 0$
    and $\delta'(s_0', w) = \delta(s_0, w)$. So $w \in L$.
    
 \item $L \cup \{\varepsilon\} \subseteq L(\mathcal A')$.
 
   As $s_0'$ is a final state, the empty word is accepted. Now suppose
   that $w \in L \setminus\{\varepsilon\}$. Hence $\delta(s_0, w) \in F$.
   Then, as $\delta'(s_0, w) = \delta(s_0, w)$, we have $w \in L(\mathcal A')$. 
 \end{enumerate}
 
 Let $m > 0$. That the bound is sharp is demonstrated by the (unary group) language $L = a^{m-1}(a^m)^*$.
 We have $\stc(L) = m$
 and $\stc(L \cup \{\varepsilon\}) = m + 1$. Note that $L \cup \{\varepsilon\}$
 is in general not a group language anymore.
 If $\mathcal A$ is the minimal commutative automaton from~\cite{DBLP:conf/cai/Hoffmann19,Hoffmann20,GomezA08}
 it is easy to see that the above construction increases the index for each letter by one, but
 leaves the period untouched. \qed
\end{proof}

   With this, we can derive our state complexity bound
   for the combined operation of the commutative closure
   and of the shuffle closure on group languages. 
   %Note that in general
   %this combined operation does not preserves regularity,
   %as shown by $\perm(\{ab\})^{\shuffle, *} = \{ w \in \{a,b\}^* \mid |w|_a = |w|_b \}$.

%% file: proof_sketch_it_shuffle.tex
The method of proof, called \emph{state label method},
is an extension of the one used 
in~\cite{DBLP:conf/dcfs/Hoffmann20}, which also includes
a detailed motivation and intuition of this method.

In what follows, we will first give an intuitive outline of the method, geared toward our intended extension, of how to use it to recognize the commutative closure of a regular language. Then, we will show how to modify it to show our statement at hand. We will only sketch the method, and will leave out some details for the sake of the bigger picture.

The method consists in labeling the points
of $\mathbb N_0^{|\Sigma|}$
with the states of a %set of states of a 
given automaton that are reachable from the start state by all words whose Parikh image
equals the point under consideration.

As it turns out,
a word is in the commutative closure
if and only if it ends in a state labeled
by a set which contains at least
one final state.

Very roughly, the resulting labeling
of $\mathbb N_0^{|\Sigma|}$
could be thought of as a more
refined version of the Parikh map
for regular languages, and in some sense
as a blend between the well-known
powerset construction, as we label with subsets of states, and the Parikh map, as we not only indicate for each point if there is a word in the language or not, but additionally store all states we could reach by words whose Parikh image equals the point in question.

%However, in~\cite{DBLP:conf/dcfs/Hoffmann20}
%the grid $\mathbb N_0^{|\Sigma|}$
%was directly labeled 

More specifically, let
$\mathcal A = (\Sigma, Q, \delta, q_0, F)$
be an automaton.
In~\cite{DBLP:conf/dcfs/Hoffmann20},
the point $p \in \mathbb N_0^{|\Sigma|}$
was labeled by the set
\[
 S_p = \{ \delta(q_0, u) \mid \psi(u) = p \}
\]
and the following holds true:
$
 v \in \perm(L(\mathcal A))
 \Leftrightarrow 
 S_{\psi(v)} \cap F \ne \emptyset.
$

Then, along any line parallel 
to the axis, which corresponds to reading in a single fixed letter, by finiteness, the state labels are ultimately periodic.
However, for each such line, the onset of the period
and the period itself may change.
For example, take the automaton with state set $Q = \{q_0, q_1, q_2\}$
over $\Sigma = \{a,b\}$ and
transition function, for $q \in Q$
and $x \in \Sigma$,
\[
 \delta(q, x) = \left\{
 \begin{array}{ll}
  q_1 & \mbox{if } q = q_0, x = a; \\
  q_0 & \mbox{if } q = q_1, x = b; \\
  q_2 & \mbox{otherwise.}
 \end{array}
 \right.
\] % todo nutze notation reguläre ausdrücke
Then, $L(\mathcal A) = (ab)^*$
and, for $p = (p_a, p_b) \in \mathbb N_0^2$,
\[
 S_{p}
  = \left\{ 
   \begin{array}{ll}
    \{ q_0, q_2 \} & \mbox{if } p_a = p_b; \\ 
    \{ q_1, q_2 \} & \mbox{if } p_a = p_b + 1; \\
    \{ q_2 \} & \mbox{otherwise.}
   \end{array}
  \right.
\]
Let $c \in \mathbb N_0$.
Then, along the lines $\{ (p_a, p_b) \in \mathbb N_0^2 \mid p_a = c \}$, 
we have $S_{(c, c + 2)} = S_{(c,c + 1)}$
and the point $(c, c+1)$
is the earliest onset after which the state labeling $S_p$ gets periodic on this line.

However, if, for any line parallel to the axis, we can bound the onset of the period and the period itself \emph{uniformly}, i.e., independently of the line we are considering,
then the commutative closure is regular,
and moreover we can construct a recognizing
automaton with these uniform bounds.

This was shown in~\cite{DBLP:conf/dcfs/Hoffmann20}
and it was shown that for group languages, we have such uniform bounds.

Note that in our example, we do not have such a uniform bound, as the onset, for example,
for the lines going in the direction $(0,1)$ 
starting at $(c,0)$ (i.e. reading in the letter $b$) was $c+1$, i.e., it grows and is not uniformly bounded. In fact, $\perm((ab)^*) = \{ u \in \{a,b\}^* : |u|_a = |u|_b \}$ is not regular.

Up to now, the method only works for the commutative closure. So, let us now describe how to modify it such that we get an automaton
for the iterated shuffle of the commutative closure of a given automaton.

First, recall that, by Theorem~\ref{thm:perm_semiring_hom},
we have
\[
 \perm(L(\mathcal A))^{\shuffle,*}
  = \perm(L(\mathcal A)^*).
\]
The usual construction for the Kleene star associates a final state with the start state,
and this is in some sense what we are doing now.
More formally, in the state labeling, we add
the start state each time we read a final state, i.e., we have another labeling which we describe next.

Let $\Sigma = \{a_1, \ldots, a_k\}$
and $e_i = \psi(a_i) = (0,\ldots, 0, 1, 0, \ldots, 0) \in \mathbb N_0^k$ be the vector with $1$ precisely at the $i$-th position
and zero everywhere else.
If $\mathcal A = (\Sigma, Q, \delta, q_0, F)$ is an automaton,  set 
\[
T_{(0,\ldots,0)} = \{ q_0 \} \quad\mbox{and}\quad
T_p = \bigcup_{\exists i \in \{1,\ldots k\} : p = q + e_i } \delta(S_q^+, a_i) \mbox{ for } p \ne (0,\ldots,0),
\]
where
\[
 S_p^+ = \left\{
 \begin{array}{ll}
  T_p \cup \{ q_0 \} & \mbox{if } T_p \cap F \ne \emptyset; \\
  T_p                & \mbox{if } T_p \cap F = \emptyset.
 \end{array}
 \right.
\]
Then,
$
 v \in \perm(L(\mathcal A)^*) \Leftrightarrow S_p^+ \cap F \ne \emptyset \mbox{ or } v = \varepsilon.
$

Note the extra condition that checks for the empty word. This is a technicality, that surely could be omitted if $q_0 \in F$, but not in the general case. 
%This also explains why we have to add one in the final 
Please see Figure~\ref{fig:state_label_update} for a visual explanation
in the case of a binary alphabet.
\input{tikz_state_labeling}

Finally, the same sufficient condition
of regularity in terms of the new state labels $S_p^+$
could be derived as in the previous case, namely
if they are uniformly bounded in the axis-parallel 
directions, then the commutative closure is regular.

Now, the sets $T_p$ are defined by the actions
of the letters $a_i$ on previous state labels $S_q^+$.
In a similar way to which it is done in~\cite{DBLP:conf/dcfs/Hoffmann20},
for a permutation automaton, we can show
that we can find such uniform bounds.

Intuitively, the reason is that
if we always permute the state labels, they cannot
get smaller as we read in more letters.
Hence, they have to grow and eventually get periodic. Also, we can show, as we only have cycles, that after a certain number of letters have been read, we have exploited all ways
that these sets could grow, i.e., we know that after we have read a certain numbers of letters we must end up in a period, and this period could also be bounded uniformly (but of course, depending on $\mathcal A$).

% todo etwas genauer?
To be a little more quantitative here, if $L_i$
denotes the order of $a_i$,
then, for each line going in the direction $e_i$, we can
show that after at most $(|Q| - 1)L_j$ many steps we must enter
the period, and the smallest period has to divide $L_j$.
This in turn could be used to derive that an automaton
with at most
\[
 \prod_{i=1}^k ((|Q|-1)L_j + L_j)
  = |Q|^k \prod_{i=1}^k L_j
\]
many states could recognize $\perm(L(\mathcal A)^+)$.
Note that this statement is only valid
for the state labeling $S_p^+$, and hence
only applies to $\perm(L(\mathcal A)^+)$.
So, to recognize  $\perm(L(\mathcal A)^*)$,
and incorporate the additional test for the empty word, we have to add one more state.

Actually, a full formal treatment, especially the steps mentioned in the previous paragraphs, is quite involved and incorporates a detailed construction of the recognizing automaton out of the state label method and a detailed analysis of the action of the permutational letters on the state set. I refer to~\cite{DBLP:conf/dcfs/Hoffmann20}
and to the extended version of this paper, which
will appear in a special issue~\cite{jalc/Hoffmann21},
for a treatment of these issues in the context of the mere commutative closure. 
% noch wie eingelesen wird, die automaten entlang achsen?
%
% actaul formal treatment quite involved

Lastly, note that the constructions are effective, as we only have to label a bounded number of grid points
of $\mathbb N_0^k$, and the state labels are computable from the transition function of $\mathcal A$.~\qed

%% file: tikz_state_labeling.tex
\begin{figure}%[!htb]
    \centering
  %\scalebox{0.9}{
 \begin{tikzpicture}[yscale=1.1, xscale=2.25] 

      %  \node at (1, -0.5)    {($p_1$, $p_2$)};
      %  \node at (-0.1, -0.5) {($p_1-1$, $p_2$)};
      
        \node at (0,0) (0p0) {$S^+_{(p_a - 1, p_b)}$};
        \node at (1,0) (1p0) {\underline{$S^+_{(p_a, p_b)}$}};
        \node at (2,0) (2p0) {$\ldots$};
         
        \node at (0,1) (0p1) {$S^+_{(p_a - 1, p_b+1)}$};
        \node at (0,2) (0p2) {$\vdots$};

        \node at (1,1) (1p1) {$S^+_{(p_a, p_b + 1)}$};
        \node at (2,1) (2p1) {$\ldots$};  
        \node at (1,2) (1p2) {$\vdots$};  
        
      \path[->] (0p0) edge [below] node {$a$} (1p0)
                 (1p0) edge [below] node {$a$} (2p0);
                 
      \path[->] (0p0)  edge [left] node {$b$} (0p1)
                 (0p1)  edge [left] node {$b$} (0p2);
                 
      \path[->] (0p1) edge [below]  node {$a$} (1p1);
      \path[->] (1p0) edge [left]   node {$b$} (1p1);
      \path[->] (1p1) edge [left]   node {$b$} (1p2);
      \path[->] (1p1) edge [below]  node {$a$} (2p1);
    \end{tikzpicture}
   % }
       % wenn man das in \caption schreibt, wird es nicht angezeigt...
  \begin{align}
     T_{(p_a, p_b+1)}  & = \delta(S^+_{(p_a-1, p_b+1)}, a) \cup \delta(S^+_{(p_a, p_b)}, b) \label{eqn:T} \\
      S^+_{(p_a, p_b+1)} & = \left\{
  \begin{array}{ll}
     T_{(p_a, p_b+1)} \cup \{s_0\} & \mbox{if } T_{(p_a, p_b+1)} \cap F \ne \emptyset; \label{eqn:S} \\ 
    T_{(p_a, p_b+1)}               & \mbox{otherwise,} 
  \end{array}
  \right.
  \end{align}
  
   \caption{Illustration of how state labels are updated for the iterated shuffle
   if new input symbols are read
   with $\Sigma = \{a,b\}$.
   For the state label $S_{(p_a, p_b)}$, after reading the letter $b$,
   we will end up at $S_{(p_a, p_b + 1)}$ and the state label is updated according
   to Equation~\eqref{eqn:T} and Equation~\eqref{eqn:S}. Seen from the state label $S_{(p_a-1, p_b)}$, we
   account for both paths given by the words $ab$ and $ba$ when ending at $(p_a, p_b+1)$,
   hence the union in the definition of $T_{(p_a, p_b + 1)}$.}
    \label{fig:state_label_update}
\end{figure}

%% file: proof_it_shuffle.tex
 Let $\Sigma = \{a_1,\ldots, a_k\}$ and $\mathcal A = (\Sigma, Q, \delta, s_0, F)$ be a permutation automaton.
 Denote by $\sigma_{\mathcal A,+} : \mathbb N_0^k \to \mathcal P(Q)$
 the state label map from Definition~\ref{def:state_label_it_shuffle}
 and by $\psi : \Sigma^* \to \mathbb N_0^k$ the Parikh map.
 By Proposition~\ref{prop:it_shuffle_state_label_map_new} we have
 $$
  \perm(L(\mathcal A))^{\shuffle,\ast} = \psi^{-1}(\sigma_{\mathcal A,+}^{-1}( \mathcal F )) \cup \{\varepsilon\}
 $$
 with $\mathcal F =  \{ S \subseteq Q \mid S \cap F \ne \emptyset\} \subseteq \mathcal P(Q)$.
 Inspecting Definition~\ref{def:state_label_it_shuffle}, we see
 that the state label map is compatible with $\mathcal A$.
 So, by Proposition~\ref{prop:state_label_for_perm_aut}, the indices 
 of the automata $\mathcal A_p^{(j)}$ from Definition~\ref{def:sequ_grid_decomp_aut}
 are universally bounded by $(|Q|-1)L_j$ and the periods divide $L_j$.
 Hence, applying Theorem~\ref{thm:regularity_condition} 
 gives\footnote{The set $\psi^{-1}(\sigma_{\mathcal A,+}^{-1}( \mathcal F ))$
 equals $\perm(L(\mathcal A))^{\shuffle,+}$. This is not explicitly stated
 but could be extracted from the proof of Proposition~\ref{prop:it_shuffle_state_label_map_new}.}
 $$
  \stc(\psi^{-1}(\sigma_{\mathcal A,+}^{-1}( \mathcal F ))) \le |Q|^k \prod_{j=1}^k L_j.
 $$
 Finally, using Lemma~\ref{lem:sc_adding_empty_word}
 gives the result for the iterated shuffle.~\qed

%% file: proof_n_ary_shuffle_prelim.tex
   For notational simplicity, we only do the case $n = 2$,
   the general case works the same way. In fact, the general case
   is only a notational complication, but nothing more.
   
   We use the general scheme for the application of the state label method
   as outlined at the start of this appendix.
   First, we have to define a state label map.
    
   \begin{definition}  %hier automaton und keien semi-atuoamton, todo bisschen was dazu sagen
   \label{def:state_label_shuffle}
      Let $\mathcal A = (\Sigma, Q_A, \delta_A, s_A, F_A)$
      and $\mathcal B = (\Sigma, Q_B, \delta_B, s_B, F_B)$
      be finite automata with disjoint state sets, i.e., 
      $Q_A \cap Q_B = \emptyset$. 
      Denote by $\sigma_{\mathcal A, \mathcal B} : \mathbb N_0^k \to \mathcal P(Q_A \cup Q_B)$ the state label function  given by
      $f : P(Q_A \cup Q_B) \times \Sigma \to P(Q_A \cup Q_B)$,
      where
      \begin{equation}\label{eq:def_state_label_shuffle}
      f(S, a) = \left\{ 
       \begin{array}{ll}
        \delta_A(S \cap Q_A, a) \cup \delta_B(S \cap Q_B, a) \cup \{ s_B \} & \mbox{if } \delta_A(S \cap Q_A, a) \cap F_A \ne \emptyset; \\ 
        \delta_A(S \cap Q_A, a) \cup \delta_B(S \cap Q_B, a) & \mbox{otherwise;}
       \end{array}\right.
      \end{equation}
      for $S \subseteq Q_A \cup Q_B$, $a \in \Sigma$, 
      %and $\sigma_{\mathcal A, \mathcal B}(0,\ldots, 0) = \{s_A\}$ if $s_A \notin F_A$,
      %and $\sigma_{\mathcal A, \mathcal B}(0,\ldots, 0) = \{s_A, s_B\}$ if $s_A \in F_A$.
      and  $
 \sigma_{\mathcal A, \mathcal B}(p) = \left\{ 
 \begin{array}{ll}
  \{ s_A, s_B \} & \mbox{ if } s_A \in F_A; \\ 
  \{ s_A \}      & \mbox{ otherwise. } 
 \end{array}
 \right.
 $
   \end{definition}

   The requirement $Q_A \cap Q_B = \emptyset$
   in most statements of this section is not a limitation,
   as we could always construct an isomorphic copy of any one of the involved automata
   if this is not fullfilled. It is more a technical requirement of the constructions, to not mix up
   what is read up to some point.

   \begin{lemma}
   \label{lem:shuffle_no_FB_state}
      Let $p \in \mathbb N_0^k$ and $\mathcal A = (\Sigma, Q_A, \delta_A, s_A, F_A)$,
      $\mathcal B = (\Sigma, Q_B, \delta_B, s_B, F_B)$ be finite automata with disjoint state sets.
      Denote by $\sigma_{\mathcal A, \mathcal B} : \mathbb N_0^k \to \mathcal P(Q_A \cup Q_B)$
      the state label map from Definition~\ref{def:state_label_shuffle}.
      If for all $q \in \mathbb N_0^k$ with $q \le p$
      we have $\sigma_{\mathcal A, \mathcal B}(q) \cap F_A = \emptyset$,
      then $\sigma_{\mathcal A, \mathcal B}(p) \cap Q_B = \emptyset$.
   \end{lemma}
    \begin{proof} 
    For $p = (0,\ldots, 0)$ the claim follows by Definition~\ref{def:state_label_shuffle}.
    Suppose $p \ne (0,\ldots, 0)$.
    Then
    $$
     \sigma_{\mathcal A, \mathcal B}(p) = \bigcup_{ p = q + \psi(b) } f(\sigma_{\mathcal A, \mathcal B}(q), b).
    $$
    By assumption $\sigma_{\mathcal A, \mathcal B}(p) \cap F_A = \emptyset$.
    Hence, for $q \in \mathbb N_0^k$ and $b \in \Sigma$ with $p = q + \psi(b)$,
    we have $f(\sigma_{\mathcal A, \mathcal B}(q), b) \cap F_A = \emptyset$.
    By Definition~\ref{def:state_label_shuffle},
    $\delta_A(\sigma_{\mathcal A, \mathcal B}(q) \cap Q_A, b) \subseteq f(\sigma_{\mathcal A, \mathcal B}(q), b)$,
    so that $\delta_A(\sigma_{\mathcal A, \mathcal B}(q) \cap Q_A, b) \cap F_A = \emptyset$.
    Again, by Definition~\ref{def:state_label_shuffle}, then
    \begin{align*}
        f(\sigma_{\mathcal A, \mathcal B}(q), b) = \delta_A(\sigma_{\mathcal A, \mathcal B}(q) \cap Q_A, b) \cup 
    \delta_B(\sigma_{\mathcal A, \mathcal B}(q) \cap Q_B, b)
    \end{align*}
    Inductively, we can assume $\sigma_{\mathcal A, \mathcal B}(q) \cap Q_B = \emptyset$.
    So the above set equals
    $$
     \delta_A(\sigma_{\mathcal A, \mathcal B}(q) \cap Q_A, b),
    $$
    which is contained in $Q_A$. Hence, as this holds for any $q \in \mathbb N_0^k$
    and $b \in \Sigma$ with $p = q + \psi(b)$, we have
    $$
      \sigma_{\mathcal A, \mathcal B}(p) = \bigcup_{ p = q + \psi(b) } f(\sigma_{\mathcal A, \mathcal B}(q), b) \subseteq Q_A
    $$
    which is equivalent with $\sigma_{\mathcal A, \mathcal B}(p) \cap Q_B = \emptyset$. \qed
   \end{proof}
   
   We will also need the following stronger version of Lemma~\ref{lem:shuffle_no_FB_state}.

\begin{lemma}
\label{lem:shuffle_QB_states}
  Let $p \in \mathbb N_0^k$ and $\mathcal A = (\Sigma, Q_A, \delta_A, s_A, F_A)$,
  $\mathcal B = (\Sigma, Q_B, \delta_B, s_B, F_B)$ be finite automata with disjoint state sets.
  Denote by $\sigma_{\mathcal A, \mathcal B} : \mathbb N_0^k \to \mathcal P(Q_A \cup Q_B)$
  the state label map from Definition~\ref{def:state_label_shuffle}.
  %If for all $q \in \mathbb N_0^k$ with $q \le p$
  %we have $\sigma_{\mathcal A, \mathcal B}(q) \cap F_A = \emptyset$,
  %then $\sigma_{\mathcal A, \mathcal B}(p) \cap Q_B = \emptyset$.
  Then
  $$
   \sigma_{\mathcal A, \mathcal B}(p) \cap Q_B = \bigcup_{\substack{\sigma_{\mathcal A, \mathcal B}(q) \cap F_A \ne \emptyset \\ p = q + \psi(u)}} \delta_B(\{s_B\}, u).
  $$
\end{lemma}
\begin{proof}
 If $p = (0,\ldots,0)$, then 
 $$
 \sigma_{\mathcal A, \mathcal B}(p) = \left\{ 
 \begin{array}{ll}
  \{ s_A, s_B \} & \mbox{ if } s_A \in F_A; \\ 
  \{ s_A \}      & \mbox{ otherwise. } 
 \end{array}
 \right.
 $$
 Hence, as then $p = q + \psi(u)$ implies $q = p = (0,\ldots,0)$ and $u = \varepsilon$, so that $\delta(\{s_B\},u) = \{s_B\}$, we have
 $$
   \sigma_{\mathcal A, \mathcal B}(p) \cap Q_B = \left\{ 
    \begin{array}{ll}
     \delta(\{ s_B \}, u) & \mbox{ if } \sigma_{\mathcal A, \mathcal B}(p) \cap F_A \ne \emptyset; \\
     \emptyset            & \mbox{ otherwise.}
    \end{array}
   \right.
 $$
 So, the equation holds. If $p \ne (0,\ldots, 0)$, then we can reason inductively,
 \begin{align*} 
  \sigma_{\mathcal A, \mathcal B}(p) \cap Q_B 
    = \left( \bigcup_{\substack{(q,b) \\ p = q + \psi(b)}} f(\sigma_{\mathcal A, \mathcal B}(q), b) \right) \cap Q_B 
    = \bigcup_{\substack{(q,b) \\ p = q + \psi(b)}} ( f(\sigma_{\mathcal A, \mathcal B}(q), b) \cap Q_B ) 
 \end{align*}
%   & = \bigcup_{\substack{(q,b) \\ p = q + \psi(b)}} \left\{ 
%     \begin{array}{ll} 
%      \delta_B(\sigma_{\mathcal A, \mathcal B}(q) \cap Q_B, b) \cup \{s_B\} & \mbox{ if } \delta_A(\sigma_{\mathcal A, \mathcal B}(q) \cap Q_A, b) \cap F_A \ne \emptyset, \\ 
%      \delta_B(\sigma_{\mathcal A, \mathcal B}(q) \cap Q_B, b) & \mbox{ otherwise. }
%     \end{array}\right. \\ 
%   & = \bigcup_{\substack{(q,b) \\ p = q + \psi(b)}} \left\{ 
%     \begin{array}{ll} 
%      \delta_B(\bigcup_{\substack{\sigma_{\mathcal A, \mathcal B}(q) \cap F_A \ne \emptyset \\ p = q + \psi(u)}} \delta(\{s_B\}, u), b) \cup \{s_B\} & \mbox{ if } \delta_A(\sigma_{\mathcal A, \mathcal B}(q) \cap Q_A, b) \cap F_A \ne \emptyset, \\ 
%      \delta_B(\bigcup_{\substack{\sigma_{\mathcal A, \mathcal B}(q) \cap F_A \ne \emptyset \\ p = q + \psi(u)}} \delta(\{s_B\}, u), b) & \mbox{ otherwise. }
%     \end{array}\right.
 By Equation~\eqref{eq:def_state_label_shuffle}, the set $f(\sigma_{\mathcal A, \mathcal B}(q), b) \cap Q_B$
 equals
 $$
   \left\{ 
     \begin{array}{ll} 
      \delta_B(\sigma_{\mathcal A, \mathcal B}(q) \cap Q_B, b) \cup \{s_B\} & \mbox{ if } \delta_A(\sigma_{\mathcal A, \mathcal B}(q) \cap Q_A, b) \cap F_A \ne \emptyset; \\ 
      \delta_B(\sigma_{\mathcal A, \mathcal B}(q) \cap Q_B, b) & \mbox{ otherwise. }
     \end{array}\right.
 $$
 By induction hypothesis, we can assume
 $$
  \sigma_{\mathcal A, \mathcal B}(q) \cap Q_B = \bigcup_{\substack{\sigma_{\mathcal A, \mathcal B}(r) \cap F_A \ne \emptyset \\ q = r + \psi(u)}} \delta(\{s_B\}, u).
 $$
 Hence
 $$
   \delta_B(\sigma_{\mathcal A, \mathcal B}(q) \cap Q_B, b)
    =  \bigcup_{\substack{\sigma_{\mathcal A, \mathcal B}(r) \cap F_A \ne \emptyset \\ q = r + \psi(u)}}
    \delta_B(\delta_B(\{s_B\}, u), b) 
 $$
 If for all $b \in \Sigma$ and $q \in \mathbb N_0^k$ with $p = q + \psi(b)$ we
 have $\delta_A(\sigma_{\mathcal A, \mathcal B}(q) \cap Q_A, b) \cap F_A = \emptyset$, then 
 by combining the above equations
 \begin{align*} 
  \sigma_{\mathcal A, \mathcal B}(p) 
   & =  \bigcup_{\substack{(q,b) \\ p = q + \psi(b)}}  
   \bigcup_{\substack{\sigma_{\mathcal A, \mathcal B}(r) \cap F_A \ne \emptyset \\ q = r + \psi(u)}}
    \delta_B(\delta_B(\{s_B\}, u), b) \\
   & = \bigcup_{\substack{(q,b) \\ p = r + \psi(u) + \psi(b) \\ \sigma_{\mathcal A, \mathcal B}(r) \cap F_A \ne \emptyset}}  
     \delta_B(\{s_B\}, ub) \\
   & = \bigcup_{\substack{(r,w) \\ p = r + \psi(w) \\ \sigma_{\mathcal A, \mathcal B}(r) \cap F_A \ne \emptyset}}  
     \delta_B(\{s_B\}, w). 
 \end{align*}
 Otherwise
 \begin{align*} 
  \sigma_{\mathcal A, \mathcal B}(p) 
   & =  \bigcup_{\substack{(q,b) \\ p = q + \psi(b)}}  
   \left( \bigcup_{\substack{\sigma_{\mathcal A, \mathcal B}(r) \cap F_A \ne \emptyset \\ q = r + \psi(u)}}
    \delta_B(\delta_B(\{s_B\}, u), b) \right) \cup \{ s_B \} \\
   & = \bigcup_{\substack{(q,b) \\ p = r + \psi(u) + \psi(b) \\ \sigma_{\mathcal A, \mathcal B}(r) \cap F_A \ne \emptyset}}  
     \delta_B(\{s_B\}, ub) \cup \{s_B\}\\
   & = \bigcup_{\substack{(r,w) \\ p = r + \psi(w) \\ \sigma_{\mathcal A, \mathcal B}(r) \cap F_A \ne \emptyset}}  
     \delta_B(\{s_B\}, w) 
 \end{align*}
 where the last equation holds, as $\{s_B\} = \delta(\{s_B\}, \varepsilon)$
 and $\sigma_{\mathcal A, \mathcal B}(p) \cap F_A \ne \emptyset$,
 so that $(p, \varepsilon)$ is part of the union. So, by induction,
 the equation from the lemma holds true. \qed
\end{proof}

   With Lemma~\ref{lem:shuffle_no_FB_state}, we can derive a connection
   between the Parikh image of $L(\mathcal A)L(\mathcal B)$
   and the state label map.
   
   \begin{proposition}
   \label{prop:parikh_image_concat}
    Suppose we have finite automata 
    $\mathcal A = (\Sigma, Q_A, \delta_A, s_A, F_A)$
    and $\mathcal B = (\Sigma, Q_B, \delta_B, s_B, F_B)$
    with $Q_A \cap Q_B = \emptyset$. %das wirklich notwendig?
    Then
    $$
     \psi( L(\mathcal A) L(\mathcal B) ) = \sigma_{\mathcal A, \mathcal B}^{-1}(\{ S \subseteq Q_A \cup Q_B \mid S \cap F_B \ne \emptyset \}).
    $$
   \end{proposition}
    \begin{proof}
    By assumption $Q_A \cap Q_B = \emptyset$. Set $Q = Q_A \cup Q_B$.
  Construct the semi-automaton $\mathcal C = (\Sigma, Q, \delta)$
   with
   $$
    \delta(q, x) = \left\{
    \begin{array}{ll}
     \delta_A(q,x) & \mbox{if } q \in Q_A; \\
     \delta_B(q,x) & \mbox{if } q \in Q_B.
    \end{array}
    \right.
   $$
   Then $\delta(S, a) = \delta_A(S\cap Q_A, a) \cup \delta_B(S\cap Q_B, a)$
   for each $S \subseteq Q$.
   Let $f : \mathcal P(Q) \times \Sigma \to \mathcal P(Q)$
   be the function from Definition~\ref{def:state_label_shuffle}
   and $\sigma_{\mathcal A, \mathcal B} : \mathbb N_0^k \to \mathcal P(Q)$
   the corresponding state label map.
   Then, for each $S \subseteq Q_A \cup Q_B$ and $a \in \Sigma$, we
   have $\delta(S, a) \subseteq f(S, a)$ by Equation~\eqref{eq:def_state_label_shuffle}, i.e., 
   the semi-automaton $\mathcal C$ is compatible with the state label map.
    
   \medskip 
   
   \noindent (i) First, let $p \in \psi(L(\mathcal A) L(\mathcal B))$.
   Then $p = \psi(u) + \psi(v)$ with $u \in L(\mathcal A)$ and $v \in L(\mathcal B)$.
   By Lemma~\ref{lem:states_along_words}, 
   $\delta(\sigma_{\mathcal A, \mathcal B}(0,\ldots, 0), u) \subseteq \sigma_{\mathcal A, \mathcal B}(\psi(u))$.
   By Definition~\ref{def:state_label_shuffle}, $\{s_A\} \subseteq \sigma_{\mathcal A, \mathcal B}(0,\ldots, 0)$.
   As $u \in L(\mathcal A)$, we have $\delta_A(s_A, u) \in F_A$.
   We will show that this implies $\{ s_B \}\subseteq \sigma_{\mathcal A, \mathcal B}(\psi(u))$.
   
   \medskip 
   
   \noindent\underline{Claim:} $\{ s_B \}\subseteq \sigma_{\mathcal A, \mathcal B}(\psi(u))$ for $u \in L(\mathcal A)$.
   \begin{quote}
      \textit{Proof of the claim.}    If $|u| = 0$, then $s_A \in F_A$.
   Hence, by Definition~\ref{def:state_label_shuffle}, 
   $\{s_A, s_B\}  = \sigma_{\mathcal A, \mathcal B}(0, \ldots, 0) = \sigma_{\mathcal A, \mathcal B}(\psi(u))$.
   Otherwise, write $u = wa$ for some $a \in \Sigma$, $w \in \Sigma^*$ and set $S = f(\sigma_{\mathcal A, \mathcal B}(0, \ldots, 0), w)$.
   %Then $\psi(u) = \psi(w) + \psi(a)$, and so by 
   So,
   $$
    f( \sigma_{\mathcal A, \mathcal B}(0, \ldots, 0), u )
     = f(S, a)
   $$
   by the extension of $f : \mathcal P(Q) \times \Sigma \to \mathcal P(Q)$ to words.
   %By Lemma~\ref{lem:states_along_words}, 
   As $\mathcal C$ is compatible with $\sigma_{\mathcal A, \mathcal B}$, we find
   $$\delta(\sigma_{\mathcal A, \mathcal B}(0,\ldots, 0), w) \subseteq S.$$
   As $\{s_A\} \subseteq \sigma_{\mathcal A, \mathcal B}(0,\ldots, 0)$,
   this gives, by construction of $\mathcal C$, then $\delta_A(\{s_A\}, w) \subseteq S$.
   Hence, $\delta_A(S \cap Q_A, a) \cap F_A \ne \emptyset$.
   But then, by Equation~\eqref{eq:def_state_label_shuffle},
   $$
    f(S, a) = \delta(S, a) \cup \{s_B\}.
   $$
   By Lemma~\ref{lem:state_label_ind_form}, and as $\psi(u) = q + \psi(w)$
   for $|u| = |w|$ implies $q = (0,\ldots, 0)$,
   $$ 
   \sigma_{\mathcal A, \mathcal B}(\psi(u)) 
    = \bigcup_{\substack{(q, w) \in \mathbb N_0^k \times \Sigma^{|u|} \\ \psi(u) = q + \psi(w)}} f(\sigma_{\mathcal A, \mathcal B}(0,\ldots,0), w). 
   $$
   Hence $f(S,a) = f( \sigma_{\mathcal A, \mathcal B}(0, \ldots, 0), u ) \subseteq \sigma_{\mathcal A, \mathcal B}(\psi(u))$
   and we can deduce $\{s_B\}\subseteq \sigma_{\mathcal A, \mathcal B}(\psi(u))$.~$\qed$
   \end{quote}
   
   \noindent Using Lemma~\ref{lem:state_label_ind_form},
   we find
   $$
     f(\sigma_{\mathcal A, \mathcal B}(\psi(u)), v) \subseteq \sigma_{\mathcal A, \mathcal B}(\psi(u) + \psi(v)).
   $$
   As 
   $$
   \delta_B(\{s_B\}, v) \subseteq \delta_B(\sigma_{\mathcal A, \mathcal B}(\psi(u)) \cap Q_B, v) 
    \subseteq \delta(\sigma_{\mathcal A, \mathcal B}(\psi(u)), v) \subseteq f(\sigma_{\mathcal A, \mathcal B}(\psi(u)), v)
   $$
    and $\delta_B(s_B, v) \in F_B$, we find
    $\sigma_{\mathcal A, \mathcal B}(p) \cap F_B \ne \emptyset$.
   This shows $\psi( L(\mathcal A) L(\mathcal B) ) \subseteq \sigma_{\mathcal A, \mathcal B}^{-1}(\{ S \subseteq Q_A \cup Q_B \mid S \cap F_B \ne \emptyset \})$.
   
   \medskip 
   
     \noindent (ii) Conversely, assume $F_B \cap  \sigma_{\mathcal A, \mathcal B}(p) \ne \emptyset$.
     
     \medskip
     
     \noindent\underline{Claim:} For each $S \subseteq Q$ and $w \in \Sigma^*$
     \begin{equation}\label{eq:Q_A}
      f(S, w) \cap Q_A = \delta_A(S \cap Q_A, w).
     \end{equation}
     \begin{quote}
         \textit{Proof of the claim.} If $|w| = 0$, then $f(S, w) \cap Q_A = S \cap Q_A = \delta(S \cap Q_A, w)$ by definition of the extension of $f$ and the transition function to words.
     Otherwise, write $w = w'a$ with $w' \in \Sigma^*$ and $a \in \Sigma$.
     Then $f(S, w'a) = f(f(S, w'), a)$.
     By Equation~\eqref{eq:def_state_label_shuffle},
     in either case $\delta_A(f(S, w') \cap Q_A, a) \cap F_A \ne \emptyset$ or 
     $\delta_A(f(S, w') \cap Q_A, a) \cap F_A = \emptyset$,
     we have
     $$
      f(f(S, w'), a) \cap Q_A = \delta_A(f(S, w')\cap Q_A, a).
     $$
     Inductively, $f(S, w')\cap Q_A = \delta_A(S\cap Q_A, w')$,
     so that $f(S, w) = \delta_A(f(S, w')\cap Q_A, a) = \delta_A(\delta_A(S\cap Q_a, w'), a) = \delta_A(S \cap Q_A, w)$. $\qed$
     \end{quote}

     \noindent Then Lemma~\ref{lem:state_label_ind_form} and Equation~\eqref{eq:Q_A} give,
     for any $q \in \mathbb N_0^k$,
     \begin{align*} 
      \sigma_{\mathcal A, \mathcal B}(q) \cap Q_A
       & = \left( \bigcup_{\psi(w) = q} f(\sigma_{\mathcal A, \mathcal B}(0,\ldots, 0), w) \right) \cap Q_A \\ 
       & = \bigcup_{\psi(w) = q} ( f(\sigma_{\mathcal A, \mathcal B}(0,\ldots, 0), w) \cap Q_A ) \\ 
       & = \bigcup_{\psi(w) = q} \delta_A(\sigma_{\mathcal A, \mathcal B}(0,\ldots, 0) \cap Q_A, w).
     \end{align*}
    %  As we assumed $F_B \cap  \sigma_{\mathcal A, \mathcal B}(p) \ne \emptyset$,
    %  by Lemma~\ref{lem:shuffle_no_FB_state} (or Lemma~\ref{lem:shuffle_QB_states}),
    %  we must have some $q \le p$ such that $\sigma_{\mathcal A, \mathcal B}(q) \cap F_A \ne \emptyset$.
    %  By the above equations, we find $w \in \Sigma^*$ with $\psi(w) = q$
    %  and $\delta_A(\sigma_{\mathcal A, \mathcal B}(0,\ldots, 0) \cap Q_A, w) \cap F_A \ne \emptyset$.
    %  As, by Equation~\eqref{eq:def_state_label_shuffle},
    %  $\sigma_{\mathcal A, \mathcal B}(0,\ldots, 0) \cap Q_A = \{s_A\}$,
    %  this gives $w \in L(\mathcal A)$.
    %  Then, with similar inductive arguments, %todo?
    %  we can show that this implies $\{s_B\}\subseteq \sigma_{\mathcal A, \mathcal B}(q)$.
    %  As the automaton $\mathcal C$ is compatible with $\sigma_{\mathcal A, \mathcal B}$ by Lemma~\ref{lem:states_along_words},
    %  $$
    %   \delta_B(s_B, u) = \delta(s_B, u) \in \sigma_{\mathcal A, \mathcal B}(p)
    %  $$
    %  for each $u$ with $p = q + \psi(u)$.
    %  Then, also by similar arguments\todo{Lemma 3 dahingehend verstärken, dass
    %  man von $s_B$ dann ein element zu dem $Q_B$ zustand hat}, we can show that we must
    %  have $v \in L(\mathcal B)$ with $p = q + \psi(u)$.
    %  So, we have $p = \psi(w) + \psi(v)$ with $w \in L(\mathcal A)$
    %  and $v \in L(\mathcal B)$.
    %  This yields $p \in \psi(L(\mathcal A) L(\mathcal B))$. \qed
     By Lemma~\ref{lem:shuffle_QB_states}, as $\sigma_{\mathcal A, \mathcal B}(p) \cap F_B \ne \emptyset$,
     we have some $v \in L(\mathcal B)$ and $q \in \mathbb N_0^k$
     with $p = q + \psi(v)$ and $\sigma_{\mathcal A, \mathcal B}(q) \cap F_A \ne \emptyset$.
     By the above equations, we find $w \in \Sigma^*$ with $\psi(w) = q$
     and $\delta_A(\sigma_{\mathcal A, \mathcal B}(0,\ldots, 0) \cap Q_A, w) \cap F_A \ne \emptyset$.
     As, by Equation~\eqref{eq:def_state_label_shuffle},
     $\sigma_{\mathcal A, \mathcal B}(0,\ldots, 0) \cap Q_A = \{s_A\}$,
     this gives $w \in L(\mathcal A)$.
     So, we have $p = \psi(w) + \psi(v)$ with $w \in L(\mathcal A)$
     and $v \in L(\mathcal B)$.
     This yields $p \in \psi(L(\mathcal A) L(\mathcal B))$. \qed
   \end{proof}

   Hence, as $\perm(L) = \psi^{-1}(\psi(L))$ for any $L \subseteq \Sigma^*$,
   we can conclude that this state labeling could be used
   to describe the commutative closure of the concatenation,
   which,  by Theorem~\ref{thm:perm_semiring_hom}, equals $\perm(L(\mathcal A))\shuffle \perm(L(\mathcal B)$.

   \begin{corollary}
   \label{cor:parikh_image_concat}
     Suppose we have finite automata 
    $\mathcal A = (\Sigma, Q_A, \delta_A, s_A, F_A)$
    and $\mathcal B = (\Sigma, Q_B, \delta_B, s_B, F_B)$
    with $Q_A \cap Q_B = \emptyset$. 
    Then
    $$
      \perm(L(\mathcal A)L(\mathcal B)) = 
     \psi^{-1}( \sigma_{\mathcal A, \mathcal B}^{-1}(\{ S \subseteq Q_A \cup Q_B \mid S \cap F_B \ne \emptyset \})).
    $$
   \end{corollary}
   
   Constructing an appropriate automaton over $Q_A \cup Q_B$ and 
   applying Theorem~\ref{thm:regularity_condition}
   then gives the next result.

%% file: proof_n_ary_shuffle.tex
As said, we only do the case $n = 2$.
    Let $\mathcal A = (\Sigma, Q_A, \delta_A, s_A, F_A)$
    and $\mathcal B = (\Sigma, Q_B, \delta_B, s_B, F_B)$
    be finite permutation automata.
    Suppose $L_j$ and $K_j$ denote the order of the letter $a_j$ 
    viewed as a permutation on $Q_A$ and $Q_B$ respectively.
    We show that
    $$
     \stc(\perm(L(\mathcal A)) \shuffle \perm(L(\mathcal B)))
      \le (Q_A + Q_B)^k \prod_{j=1}^k  \lcm(L_j, K_j). 
    $$   % wenn beide automaten als ein semi-automat aufgefasst, dann disjunkte zyklen
    
%  \begin{comment}
%   Note that the state label function from Definition~\ref{def:state_label_shuffle}
%   is automata induced if we construct the semi-automaton $\mathcal C = (\Sigma, Q_A \cup Q_B, \delta, s_A)$
%   with
%   $$
%     \delta(q, x) = \left\{
%     \begin{array}{ll}
%      \delta_A(q,x) & \mbox{if } q \in Q_A; \\
%      \delta_B(q,x) & \mbox{if } q \in Q_B.
%     \end{array}
%     \right.
%   $$
%   We will use this fact without explicit
%   mentioning of the above construction.
%   Also note that the requirement $Q_A \cap Q_B = \emptyset$
%   in most statements of this section is not a limitation,
%   as we could always construct an isomorphic copy of any one of those automata
%   if this is not fullfilled. It is more a technical requirement of the constructions, to not mix up
%   what is read up to some point.
% \end{comment}   

  We can assume $Q_A \cap Q_B = \emptyset$. Set $Q = Q_A \cup Q_B$.
  Construct the semi-automaton $\mathcal C = (\Sigma, Q, \delta)$
   with
   $$
    \delta(q, x) = \left\{
    \begin{array}{ll}
     \delta_A(q,x) & \mbox{if } q \in Q_A; \\
     \delta_B(q,x) & \mbox{if } q \in Q_B.
    \end{array}
    \right.
   $$
   Let $f : \mathcal P(Q) \times \Sigma \to \mathcal P(Q)$
   be the function from Definition~\ref{def:state_label_shuffle}
   and $\sigma_{\mathcal A, \mathcal B} : \mathbb N_0^k \to \mathcal P(Q)$
   the corresponding state label map.
   Then, for each $S \subseteq Q_A \cup Q_B$ and $a \in \Sigma$, we
   have $\delta(S, a) \subseteq f(S, a)$ by Equation~\eqref{def:state_label_shuffle}, i.e., 
   the semi-automaton $\mathcal C$ is compatible with the state label map.
   The automaton $\mathcal C$ is a permutation semi-automaton,
   and each letter $a_j \in \Sigma$ has order $\lcm(L_j, K_j)$, viewed
   as a permutation on $Q_A \cup Q_B$.
   By Proposition~\ref{prop:state_label_for_perm_aut},
   the automata $\mathcal A_p^{(j)}$ from Definition~\ref{def:sequ_grid_decomp_aut}
   have index at most $(|Q_A \cup Q_B| - 1)\lcm(L_j, K_j)$
   and period at most $\lcm(L_j, K_j)$.
   Hence, using Theorem~\ref{thm:regularity_condition},
   the language $\psi^{-1}(\sigma_{\mathcal A, \mathcal B}(\mathcal F))$
   with $\mathcal F = \{ S \subseteq Q_A \cup Q_B \mid S \cap F_B \ne \emptyset \}$
   is accepted by an automaton of size at most
   $$
    \prod_{j=1}^k \bigg( (|Q_A \cup Q_B| - 1)\lcm(L_j, K_j) + \lcm(L_j, K_j) \bigg).
   $$
   By Corollary~\ref{cor:parikh_image_concat}, the result follows.~\qed

%% file: ms.bbl
\begin{thebibliography}{Cam99}

\bibitem[Cam99]{cameron_1999}
Peter~J. Cameron.
\newblock {\em Permutation Groups}.
\newblock London Mathematical Society Student Texts. Cambridge University
  Press, 1999.

\bibitem[GA08]{GomezA08}
A.~Cano G{\'{o}}mez and G.~I. Alvarez.
\newblock Learning commutative regular languages.
\newblock In Alexander Clark, Fran{\c{c}}ois Coste, and Laurent Miclet,
  editors, {\em {ICGI} 2008, Saint-Malo, France, September 22-24, 2008,
  Proceedings}, volume 5278 of {\em LNCS}, pages 71--83. Springer, 2008.

\bibitem[Hof]{Hoffmann20}
S.~Hoffmann.
\newblock State complexity, properties and generalizations of commutative
  regular languages.
\newblock {\em Information and Computation}, (submitted).

\bibitem[Hof19]{DBLP:conf/cai/Hoffmann19}
Stefan Hoffmann.
\newblock Commutative regular languages - properties and state complexity.
\newblock In Miroslav Ciric, Manfred Droste, and Jean{-}{\'{E}}ric Pin,
  editors, {\em Algebraic Informatics - 8th International Conference, {CAI}
  2019, Ni{\v{s}}, Serbia, June 30 - July 4, 2019, Proceedings}, volume 11545
  of {\em Lecture Notes in Computer Science}, pages 151--163. Springer, 2019.

\bibitem[Hof20]{DBLP:conf/dcfs/Hoffmann20}
Stefan Hoffmann.
\newblock State complexity bounds for the commutative closure of group
  languages.
\newblock In Galina Jir{\'{a}}skov{\'{a}} and Giovanni Pighizzini, editors,
  {\em Descriptional Complexity of Formal Systems - 22nd International
  Conference, {DCFS} 2020, Vienna, Austria, August 24-26, 2020, Proceedings},
  volume 12442 of {\em LNCS}, pages 64--77. Springer, 2020.

\end{thebibliography}


\begin{thebibliography}{10}
\providecommand{\url}[1]{\texttt{#1}}
\providecommand{\urlprefix}{URL }
\providecommand{\doi}[1]{https://doi.org/#1}

\bibitem{DBLP:journals/iandc/BouajjaniMT07}
Bouajjani, A., Muscholl, A., Touili, T.: Permutation rewriting and algorithmic
  verification. Inf. Comput.  \textbf{205}(2),  199--224 (2007)

\bibitem{DBLP:journals/iandc/BrodaMMR18}
Broda, S., Machiavelo, A., Moreira, N., Reis, R.: Automata for regular
  expressions with shuffle. Inf. Comput.  \textbf{259}(2),  162--173 (2018)

\bibitem{BrzozowskiJLRS16}
Brzozowski, J.A., Jir{\'{a}}skov{\'{a}}, G., Liu, B., Rajasekaran, A.,
  Szyku{\l}a, M.: On the state complexity of the shuffle of regular languages.
  In: C{\^{a}}mpeanu, C., Manea, F., Shallit, J.O. (eds.) Descriptional
  Complexity of Formal Systems - 18th {IFIP} {WG} 1.2 International Conference,
  {DCFS} 2016, Bucharest, Romania, July 5-8, 2016. Proceedings. LNCS,
  vol.~9777, pp. 73--86. Springer (2016)

\bibitem{CamHab74}
Campbell, R.H., Habermann, A.N.: The specification of process synchronization
  by path expressions. In: Gelenbe, E., Kaiser, C. (eds.) Operating Systems OS.
  LNCS, vol.~16, pp. 89--102. Springer (1974)

\bibitem{DBLP:journals/jalc/CampeanuSY02}
C{\^{a}}mpeanu, C., Salomaa, K., Yu, S.: Tight lower bound for the state
  complexity of shuffle of regular languages. J. Autom. Lang. Comb.
  \textbf{7}(3),  303--310 (2002)

\bibitem{DBLP:journals/jalc/CaronLP20}
Caron, P., Luque, J., Patrou, B.: A combinatorial approach for the state
  complexity of the shuffle product. J. Autom. Lang. Comb.  \textbf{25}(4),
  291--320 (2020)

\bibitem{DBLP:journals/ita/CeceHM08}
C{\'{e}}c{\'{e}}, G., H{\'{e}}am, P., Mainier, Y.: Efficiency of automata in
  semi-commutation verification techniques. {RAIRO} Theor. Informatics Appl.
  \textbf{42}(2),  197--215 (2008)

\bibitem{DBLP:books/ws/95/DR1995}
Diekert, V., Rozenberg, G. (eds.): The Book of Traces. World Scientific (1995)

\bibitem{FerPSV2017}
Fernau, H., Paramasivan, M., Schmid, M.L., Vorel, V.: Characterization and
  complexity results on jumping finite automata. Theo. Comp. Sci.
  \textbf{679},  31--52 (2017)

\bibitem{FernauHoffmann2019}
Fernau, H., Hoffmann, S.: Extensions to minimal synchronizing words. J. Autom.
  Lang. Comb.  \textbf{24}(2-4),  287--307 (2019). \doi{10.25596/jalc-2019-287}

\bibitem{GaoMRY17}
Gao, Y., Moreira, N., Reis, R., Yu, S.: A survey on operational state
  complexity. Journal of Automata, Languages and Combinatorics  \textbf{21}(4),
   251--310 (2017)

\bibitem{GinsburgSpanier66}
Ginsburg, S., Spanier, E.H.: Bounded regular sets. Proceedings of the American
  Mathematical Society  \textbf{17},  1043--1049 (1966)

\bibitem{GinsburgSpanier66b}
Ginsburg, S., Spanier, E.H.: {Semigroups, Presburger formulas, and languages.}
  Pacific Journal of Mathematics  \textbf{16}(2),  285--296 (1966)

\bibitem{Gohon85}
Gohon, P.: An algorithm to decide whether a rational subset of n{\^{}}k is
  recognizable. Theor. Comput. Sci.  \textbf{41},  51--59 (1985)

\bibitem{GomezA08}
G{\'{o}}mez, A.C., Alvarez, G.I.: Learning commutative regular languages. In:
  Clark, A., Coste, F., Miclet, L. (eds.) {ICGI} 2008, Saint-Malo, France,
  September 22-24, 2008, Proceedings. LNCS, vol.~5278, pp. 71--83. Springer
  (2008)

\bibitem{DBLP:journals/iandc/GomezGP13}
G{\'{o}}mez, A.C., Guaiana, G., Pin, J.: Regular languages and partial
  commutations. Inf. Comput.  \textbf{230},  76--96 (2013)

\bibitem{Hoffmann20}
Hoffmann, S.: State complexity, properties and generalizations of commutative
  regular languages. Information and Computation  \textbf{(submitted)}

\bibitem{jalc/Hoffmann21}
Hoffmann, S.: State complexity bounds for the commutative closure of group
  languages. Journal of Automata, Languages and Combinatorics
  \textbf{(submitted)}

\bibitem{DBLP:conf/cai/Hoffmann19}
Hoffmann, S.: Commutative regular languages - properties and state complexity.
  In: Ciric, M., Droste, M., Pin, J. (eds.) Algebraic Informatics - 8th
  International Conference, {CAI} 2019, Ni{\v{s}}, Serbia, June 30 - July 4,
  2019, Proceedings. Lecture Notes in Computer Science, vol. 11545, pp.
  151--163. Springer (2019)

\bibitem{DBLP:conf/dcfs/Hoffmann20}
Hoffmann, S.: State complexity bounds for the commutative closure of group
  languages. In: Jir{\'{a}}skov{\'{a}}, G., Pighizzini, G. (eds.) Descriptional
  Complexity of Formal Systems - 22nd International Conference, {DCFS} 2020,
  Vienna, Austria, August 24-26, 2020, Proceedings. LNCS, vol. 12442, pp.
  64--77. Springer (2020)

\bibitem{HopUll79}
Hopcroft, J.E., Ullman, J.D.: Introduction to Automata Theory, Languages, and
  Computation. Addison-Wesley Publishing Company (1979)

\bibitem{Ito04}
Ito, M.: Algebraic Theory of Automata and Languages. World Scientific (2004)

\bibitem{DBLP:journals/tcs/Jantzen81}
Jantzen, M.: The power of synchronizing operations on strings. Theor. Comput.
  Sci.  \textbf{14},  127--154 (1981)

\bibitem{DBLP:journals/tcs/Jantzen85}
Jantzen, M.: Extending regular expressions with iterated shuffle. Theor.
  Comput. Sci.  \textbf{38},  223--247 (1985)

\bibitem{DBLP:journals/tcs/JedrzejowiczS01}
Jedrzejowicz, J., Szepietowski, A.: Shuffle languages are in {P}. Theor.
  Comput. Sci.  \textbf{250}(1-2),  31--53 (2001)

\bibitem{DBLP:conf/stoc/Kimura76}
Kimura, T.: An algebraic system for process structuring and interprocess
  communication. In: Chandra, A.K., Wotschke, D., Friedman, E.P., Harrison,
  M.A. (eds.) Proceedings of the 8th Annual {ACM} Symposium on Theory of
  Computing, May 3-5, 1976, Hershey, Pennsylvania, {USA}. pp. 92--100. {ACM}
  (1976)

\bibitem{Landau1903}
Landau, E.G.H.: {Ü}ber die {M}aximalordnung der {P}ermutationen gegebenen
  {G}rades. Archiv der Mathematik und Physik  \textbf{5}(3),  92--103 (1903)

\bibitem{Lvov73}
L'vov, M.: Commutative closures of regular semigroup languages. Kibernetika
  (Kiev)  \textbf{2},  54--58 (1973)

\bibitem{DBLP:conf/mfcs/Marzurkiewicz75}
Mazurkiewicz, A.W.: Parallel recursive program schemes. In: Becv{\'{a}}r, J.
  (ed.) {MFCS} 1975, 4th Symposium, Mari{\'{a}}nsk{\'{e}} L{\'{a}}zne,
  Czechoslovakia, September 1-5, 1975, Proceedings. {LNCS}, vol.~32, pp.
  75--87. Springer (1975)

\bibitem{DBLP:journals/ipl/MuschollP96}
Muscholl, A., Petersen, H.: A note on the commutative closure of star-free
  languages. Inf. Process. Lett.  \textbf{57}(2),  71--74 (1996)

\bibitem{DBLP:journals/jacm/Parikh66}
Parikh, R.: On context-free languages. J. {ACM}  \textbf{13}(4),  570--581
  (1966)

\bibitem{Pin86}
Pin, J.: Varieties Of Formal Languages. Plenum Publishing Co. (1986)

\bibitem{DBLP:reference/hfl/Pin97}
Pin, J.: Syntactic semigroups. In: Rozenberg, G., Salomaa, A. (eds.) Handbook
  of Formal Languages, Volume 1, pp. 679--746. Springer (1997)

\bibitem{Redko63}
Redko, V.: On the commutative closure of events. Dopovidi Akad. Nauk Urkain.
  RSR pp. 1156--1159 (1963)

\bibitem{DBLP:conf/latin/Sakarovitch92}
Sakarovitch, J.: The "last" decision problem for rational trace languages. In:
  Simon, I. (ed.) {LATIN} '92, 1st Latin American Symposium on Theoretical
  Informatics, S{\~{a}}o Paulo, Brazil, April 6-10, 1992, Proceedings. Lecture
  Notes in Computer Science, vol.~583, pp. 460--473. Springer (1992)

\bibitem{Shaw78zbMATH03592960}
{Shaw}, A.C.: {Software descriptions with flow expressions.} {IEEE Trans.
  Softw. Eng.}  \textbf{4},  242--254 (1978)

\end{thebibliography}
